\documentclass[letterpaper, journal]{IEEEtran}

\let\solution=  \iffalse 

\usepackage{amssymb}
\usepackage{amsfonts}
\usepackage{array}      
\usepackage{amsmath}
\usepackage{amsthm}
\usepackage{graphicx}
\usepackage[usenames,dvipsnames]{color}
\usepackage{cite}
\usepackage{url}

\newtheorem{theorem}{Theorem}
\newtheorem{definition}{Definition}
\newtheorem{lemma}{Lemma}
\newtheorem{corollary}{Corollary}

\def\Pcppp{P_{\rm c}^{\rm PPP}}
\def\Poppp{P_{\rm o}^{\rm PPP}}
\def\pt{p_{\rm t}}

\def\Pc{P_{\rm c}}

\def\lambdap{\lambda_{\rm p}}
\def\lambdab{\lambda_{\rm b}}
\def\rc{r_{\rm c}}
\def\rh{r_{\rm h}}
\def\Phip{\Phi_{\rm p}}
\def\Phib{\Phi_{\rm b}}
\def\PhibI{\Phi_{{\rm b}_1}}
\def\PhibII{\Phi_{{\rm b}_2}}
\def\PhibIII{\Phi_{{\rm b}_3}}

\def\Iy{I(\Phi^{\zeta}_o)}       
\def\Izn{\hat{I}^{!\zeta}_n}
\def\IzII{\hat{I}(\Phi^{\zeta})}                         
\def\np{{\rm NP}}
\def\Ks{K}

\DeclareMathOperator*{\argmin}{arg\,min}

\def\mh#1{{\iffalse\bf\textcolor{blue}{[ #1 ]}\fi}}

\newlength{\figwidth}

\solution 
\else
\fi

\ifCLASSINFOpdf
\else
\fi

\begin{document}
\title{Asymptotic Deployment Gain:\\A Simple Approach to Characterize  the SINR Distribution in General Cellular Networks}

\author{\IEEEauthorblockN{Anjin Guo and Martin Haenggi\\
Dept.~of Electrical Engineering\\
University of Notre Dame\\}\thanks{ 
The support of the NSF (grants CNS 1016742 and CCF 1216407) is gratefully acknowledged.}
}

\maketitle

\begin{abstract}
In cellular network models, the base stations are usually assumed to form a lattice or a Poisson point process (PPP). In reality,
however, they are deployed neither fully regularly nor completely randomly. Accordingly, in this paper, we consider the very
general class of motion-invariant models and analyze the behavior of the outage probability (the probability that the
signal-to-interference-plus-noise-ratio (SINR) is smaller than a threshold) as the threshold goes to zero. We show that,
remarkably, the slope of the outage probability (in dB) as a function of the threshold (also in dB) is the same for essentially all motion-invariant point processes. The slope merely depends on the fading statistics.

Using this result, we introduce the notion of the \emph{asymptotic deployment gain} (ADG), which characterizes the horizontal gap between the success probabilities of the PPP and another point process in the high-reliability regime (where the success probability is near 1).  

To demonstrate the usefulness of the ADG for the characterization of the SINR distribution,
we investigate the outage probabilities  and the ADGs for different point processes and fading statistics by simulations.
\end{abstract}


\section{Introduction}
\subsection{Motivation}
The topology of the base stations (BSs) in cellular networks depends on many natural or man-made factors,
such as the landscape, topography, bodies of water, population densities, and traffic demands. Despite,
base stations were usually modeled deterministically as triangular or square lattices until recently, when it was
shown in \cite{Tract} that a completely irregular point process, the Poisson point process (PPP) \cite{book,JSAC09}, may be
used without  loss in accuracy but significant gain in analytical tractability. Real deployments fall somewhere
in between these two extremes of full regularity (the triangular lattice) and complete randomness (the PPP),
as investigated in  \cite{AGTWC13} using base station data from the UK. They exhibit some degree of
repulsion between the BSs, since the operators do not place them closely together. Such repulsion
can be modeled using point processes with a hard minimum distance between BSs (hard-core processes)
or a high likelihood that BSs are a certain distance apart (soft-core processes). At a larger scale, at the level of a state or country, BSs may appear clustered due to the high density of BSs in cities and low density in rural regions.  The analysis of such non-Poisson processes is, however,  significantly more difficult than the analysis of the PPP,
since dependencies exist between the locations of base stations.

The signal-to-interference-plus-noise ratio (SINR) distribution is one of the main performance metrics in cellular networks. In this paper, we mainly consider the \emph{success probability}, defined as the complementary cumulative distribution function of the SINR, i.e., $\Pc(\theta) \triangleq \mathbb{P}(\textrm{SINR} > \theta)$.
The success probability is the same as the \emph{coverage probability} in many other works (e.g., \cite{Tract}), but we use ``success" instead of ``coverage", because in the  cellular industry,
``coverage"  does not include small-scale fading, but is based on  only path loss and shadowing.
The success probability  depends on many factors, such as the fading, the path loss and the distribution of the BSs.
In \cite{Tract}, the authors derive the expressions for the success probability and the mean achievable rate for networks whose BSs form a homogeneous PPP. For general models, it is much harder to compute the success probability due to the
dependence between BS locations mentioned above, and we are not aware of any tractable analytical methods that
are applicable in general.
In this paper,
we propose a novel approach to evaluate the success probability  of  cellular networks, where BSs follow non-Poisson processes. 

\subsection{Related Work}
The spatial configuration of the BSs (or transmitters) plays a critical role in the performance evaluation of cellular networks (or general wireless networks), since the SINR critically relies on the distances between BSs and users (or transmitters and receivers).  
Network performance analysis using stochastic geometry have drawn considerable attention \cite{JSAC09,Tract, modelling0, PPPcellular0,KTier,DPPCluster,GantiCluster,NOW,NSCluster,tit11,High-Reliability, MIND,AGTWC13,nonPoisson1,nonPoisson2,PoissonNonCell1,PoissonNonCell2,NaTWC14, ModelCWN}.  
Recent related works can  be roughly divided into two categories. One is based on the assumption of modeling the BSs or access points as  Poisson-based point processes (e.g., the PPP and the Poisson cluster process) in cellular networks, e.g.~\cite{Tract,PPPcellular0,KTier,DPPCluster,NSCluster}. The other one is dealing with general point processes in non-cellular networks, especially in wireless ad hoc networks, e.g. \cite{NOW,tit11, High-Reliability,MIND}. 
Of course, there are some other types of works, such as the type of using the Poisson-based point processes in non-cellular networks, e.g. \cite{PoissonNonCell1,PoissonNonCell2, GantiCluster}, and the type of using non-Poisson point processes in non-cellular networks, e.g. \cite{nonPoisson1,nonPoisson2}, but they are not closely related  this paper. Our focus is applying general point processes to cellular networks, which has seldom been studied.

Regarding the first category, in cellular networks, the PPP is advantageous for modeling the BSs configuration \cite{Tract,PPPcellular0,KTier,DPPCluster},  due to its analytical tractability. 
Poisson-based processes, especially Poisson cluster processes, e.g. in \cite{NSCluster},  have been used to model the small cell tier in heterogeneous cellular networks, where the BS tier is still modeled as the PPP.  
Non-Poisson processes, such as hard-core processes, are less studied in cellular networks, due to the absence of an analytical form for the probability generating functional and the Palm characterization of the point process distribution. 
Exceptions are the related works in \cite{NaTWC14, ModelCWN, AGTWC13}.  
In \cite{NaTWC14}, the authors applied the $\beta$-Ginibre point process, where points exhibit repulsion, in cellular networks. 
In \cite{ModelCWN}, the Geyer saturation process was used to model the real cellular  service site locations.
In \cite{AGTWC13}, we modeled BSs as  the Strauss process and the Poisson hard-core process,  and provided fitting methods using real BS data sets. 

As for the second category,  
general point processes have been used to model the transmitting nodes in non-cellular networks, see, e.g., \cite{tit11, High-Reliability,MIND}. 
In \cite{tit11}, the authors analyzed the success probability in an asymptotic regime where the density of interferers goes to $0$ in wireless networks with general fading and node distribution. 
The paper \cite{High-Reliability} provides an in-depth study of the outage probability of general ad hoc networks, where the nodes form an arbitrary motion-invariant point process, under Rayleigh fading as the density of interferers goes to $0$. 
In \cite{MIND}, the tail properties of interference for any motion-invariant spatial distribution of transmitting nodes were derived.

In this paper,  we consider a general class of point processes for modeling possible BS configurations. 
In homogeneous cellular networks, each user is usually serviced by its nearest BS, though not necessarily. 
When general point processes are applied in such networks, one of the main emerging difficulties  is that  the point process distribution conditioned on  an empty ball around the user is unknown.  Moreover, the empty space function has to be considered, resulting in the growth of the complexity.  
Tackling those difficulties directly  is seldom seen in the literature.

\subsection{Contributions}

In this paper, we provide an indirect approach to the success probability  analysis  of an
arbitrary motion-invariant (isotropic and stationary\footnote{Stationarity implies
that the success probability does not depend on the location of the typical user.}) point process \cite[Ch.~2]{book}
 by comparing its success probability to the success probability of the PPP. To validate this approach,
 we establish that {\em the outage probability $1-\Pc(\theta)$ of essentially
all motion-invariant (m.i.) point processes, expressed in dB, as a function of the SINR threshold $\theta$, also in dB, has the
same slope as $\theta\to 0$}. The slope depends on the fading statistics.
This result shows that asymptotically the success curves $\Pc(\theta)$
of all m.i.~models are just (horizontally) shifted versions of each other in a log-log plot,
and the shift can be quantified in terms of the horizontal difference $\hat G$ along the $\theta$ (in dB) axis. Since the success
probability of the PPP is known analytically, the PPP is a sensible choice as a reference model, which then allows
to express the success probability of an arbitrary m.i.~model as a gain relative to the PPP. This gain is called the
{\em asymptotic deployment gain} (ADG).

We introduced the concept of the \emph{deployment gain} (DG) in our previous work \cite{AGTWC13}. It measures how close a point process or a point set is to the PPP at a given target success probability. Here we extend the DG to include noise and then,
to obtain a quantity that does not depend on a target success probability, formally define its asymptotic counterpart---the ADG.

The paper makes the following contributions:
\begin{itemize}
\item We introduce the asymptotic deployment gain.
\item We formally prove its existence for a large class of m.i.~point processes.
\item We show how the asymptotic slope of the outage probability depends on the fading statistics.
\item We demonstrate through simulations how the ADG can be used to quantify the success probability of
several non-Poisson models, even if the SINR threshold $\theta$ is not small.
\end{itemize}

The rest of this paper is organized as follows. In Section \uppercase\expandafter{\romannumeral2}, we introduce the system model and the  ADG.
We investigate the existence of the ADG and study the asymptotic properties of the outage probability in Section \uppercase\expandafter{\romannumeral3}. 
Applications of the ADG are provided in Section \uppercase\expandafter{\romannumeral4}.
In Section \uppercase\expandafter{\romannumeral5}, we show simulation results for some specific network models. Conclusions are drawn in  Section \uppercase\expandafter{\romannumeral6}.

\section{System Model and Asymptotic Deployment Gain}
\subsection{System Model}
We consider a cellular network that consists of BSs and mobile users. The BSs are modeled as a general m.i.~point
process $\Phi$ of intensity $\lambda$ on the plane.
We assume that $\Phi$ is mixing  \cite[Def.~2.31]{book}, which implies that the second moment density  $\rho^{(2)}(x_1,x_2) \to \lambda^2$ as $\|x_1-x_2\|\to\infty$.  
Intuitively, $\rho^{(2)}(x_1,x_2)$ is the probability that there are two points of $\Phi$ at $x_1$ and $x_2$ in the infinitesimal volumes $dx_1$ and $dx_2$.  
Rigorously, it is the density pertaining to the second factorial moment measure \cite[Def.~6.4]{book},
which is given by
\begin{equation}
\alpha^{(2)}(A \times B) = \mathbb{E}\Big( \sum_{x,y \in \Phi}^{\neq} \mathbf{1}_A(x)\mathbf{1}_B(y) \Big)=\!\!
\int\limits_{A \times B} \rho^{(2)}(y-x) dxdy,  \nonumber
\end{equation}
where  $A, B$ are two compact subsets of $\mathbb{R}^2$, and the $\neq$ symbol indicates that the sum is taken only over distinct point pairs.
Since the point processes considered are m.i.,   
$\rho^{(2)}(x_1,x_2)$  only depends on $\|x_1-x_2\|$. 
Without ambiguity, we let $\rho^{(2)}(x_2-x_1) \triangleq \rho^{(2)}(x_1,x_2)$.
Similarly, the $n$th moment density $\rho^{(n)}(x_1,x_2,\ldots,x_n)$ is the  the density (with respect to the Lebesgue measure) pertaining to the $n$th-order factorial moment measure $\alpha^{(n)}$, and we let $\rho^{(n)}(x_2-x_1,\ldots,x_n-x_1) \triangleq  \rho^{(n)}(x_1,\ldots,x_n)$.

We assume all BSs are always transmitting and the transmit power is fixed to 1.   
Each mobile user receives signals from its nearest BS, and all other BSs act as interferers (the frequency reuse factor is 1).
Every signal is assumed to experience path loss and  fading.    
We consider both non-singular and singular path loss models, which are, respectively,    
$\ell(x)= (1+\|x\|^{\alpha})^{-1}$ and $\ell(x)= \|x\|^{-\alpha}$, where $\alpha > 2$. (Since $\ell(x)$ only depends on $\|x\|$, in this paper, $\ell(x)$ and $\ell(\|x\|)$ are equivalent.)
We assume that the fading is independent and identically distributed (i.i.d.) for signals from all BSs. 
The fading can be small-scale fading,  shadowing  or a combination of the two.  
We  mainly   focus on Nakagami-$m$ fading, which includes  Rayleigh fading as a special case, and the combination of Nakagami-$m$ fading and log-normal shadowing. 
The thermal noise  is assumed to be additive and constant with power $W$. 
We define the mean SNR as the received SNR at a distance of $r = 1$, where its value is $1/(2W)$ for the non-singular path loss model and $1/W$ for the singular path loss model.  

To formulate the SINR and the success probability, we first define the nearest-point operator $\np_{\varphi}$ for a  point pattern $\varphi\subset \mathbb{R}^2$ as
\begin{equation}
\np_{\varphi}(x) \triangleq \argmin_{y \in \varphi}  \{ \|y-x\| \}, \quad x \in \mathbb{R}^2.
\end{equation}
If the nearest point is not unique, the operator  picks  one of the nearest points uniformly at random.
The
SINR at location $z \in \mathbb{R}^2$  has the form
\begin{equation}
\textrm{SINR}_{z} = \frac{h_{u} \ell(u-z)}{W + \sum_{x \in \Phi \setminus \{u\} }h_x \ell(x-z) },
\end{equation}
where $u = \np_{\Phi}(z)$ and $h_x$ denotes the i.i.d. fading variable for $x \in \Phi$ with CDF $F_h$ and PDF $f_h$.
For a m.i.~point process,  the success probability $\mathbb{P} (\textrm{SINR}_{z} > \theta)$ does not depend on $z$, and we define
\begin{equation}
\Pc(\theta) = \mathbb{P} (\textrm{SINR} > \theta).
\end{equation}
Hence, without  loss of generality, we focus on the success probability at the origin $o$.
Since each user
communicates with its nearest BS, the interference at $o$ only comes
from the BSs outside the open disk $b(o,r) \triangleq \{ x\in \mathbb{R}^2: \| x\| < r \}$, where $r = \| \np_{\Phi}(o) \|$. 
The total interference, denoted by $I(\Phi)$, is
\begin{equation}
I(\Phi) = \sum_{x\in \Phi \setminus \np_{\Phi}(o)} h_x\ell(x).   
\end{equation}

\mh{This I do not understand. If $\zeta$ is a location (i.e., an arbitrary element of $\mathbb R^2$), then a.s.~$\zeta\notin\Phi$, so writing
$\Phi\setminus\{\zeta\}$ is potentially misleading. Can't you assume $\zeta\in\Phi$?}

\subsection{Asymptotic Deployment Gain}
In \cite{AGTWC13}, we introduced  the  deployment gain (DG) for interference-limited networks. 
Here we redefine the DG, to include the thermal noise. 

\begin{definition}[Deployment gain]
The deployment gain, denoted by $G(\pt)$, is the  ratio of the $\theta$ values  between the success curves  of  the given point process (or point set) and the PPP  at a given target success probability $\pt$, i.e.,
\begin{equation}
G(\pt) = \frac{\Pc^{-1}(\pt) }{(\Pcppp)^{-1}(\pt)}
\label{G11}
\end{equation}
where $\Pcppp(\theta)$ and $\Pc(\theta)$ are, respectively, the success  probabilities of the PPP and the given point process $\Phi$. 
\end{definition}

This definition is analogous to the notion of the coding gain   commonly used in coding theory \cite[Ch.~1]{CodingGain}.

Fig. \ref{f0} shows the success probability of the PPP, the Mat\'{e}rn cluster process (MCP)  \cite[Ch.~3]{book},   
and the randomly translated triangular lattice.
The intensities of all the three point processes are the same.
We observe that for $\pt > 0.3$, the DG is approximately  constant, e.g. the DG of the MCP is about $-3$ dB. In Fig. \ref{f0}, the success probability curves of the PPP that are shifted by $G(0.6)$ (in dB) of the MCP and the triangular lattice are also drawn. We see that  the shifted curves overlap quite exactly with the curves of the MCP and the triangular lattice, respectively, for
all $\pt > 0.3$.  
It is thus sensible to study the DG as $\pt \to 1$ and find out whether the DG  approaches a constant.
To do so,
analogous to the notion of the asymptotic coding gain, we  define the asymptotic deployment gain (ADG).
\begin{definition}[Asymptotic deployment gain] The ADG, denoted by $\hat{G}$, is the deployment gain $G(\pt)$ when $\theta \to 0$, or, equivalently, when  $\pt \to 1$:
\begin{equation}
\hat{G} = \lim_{\pt \to 1} G(\pt) .
\label{G12}
\end{equation}
\end{definition}

\solution
\begin{figure}[hbtp]
\centering
  \includegraphics[width=0.8\figwidth]{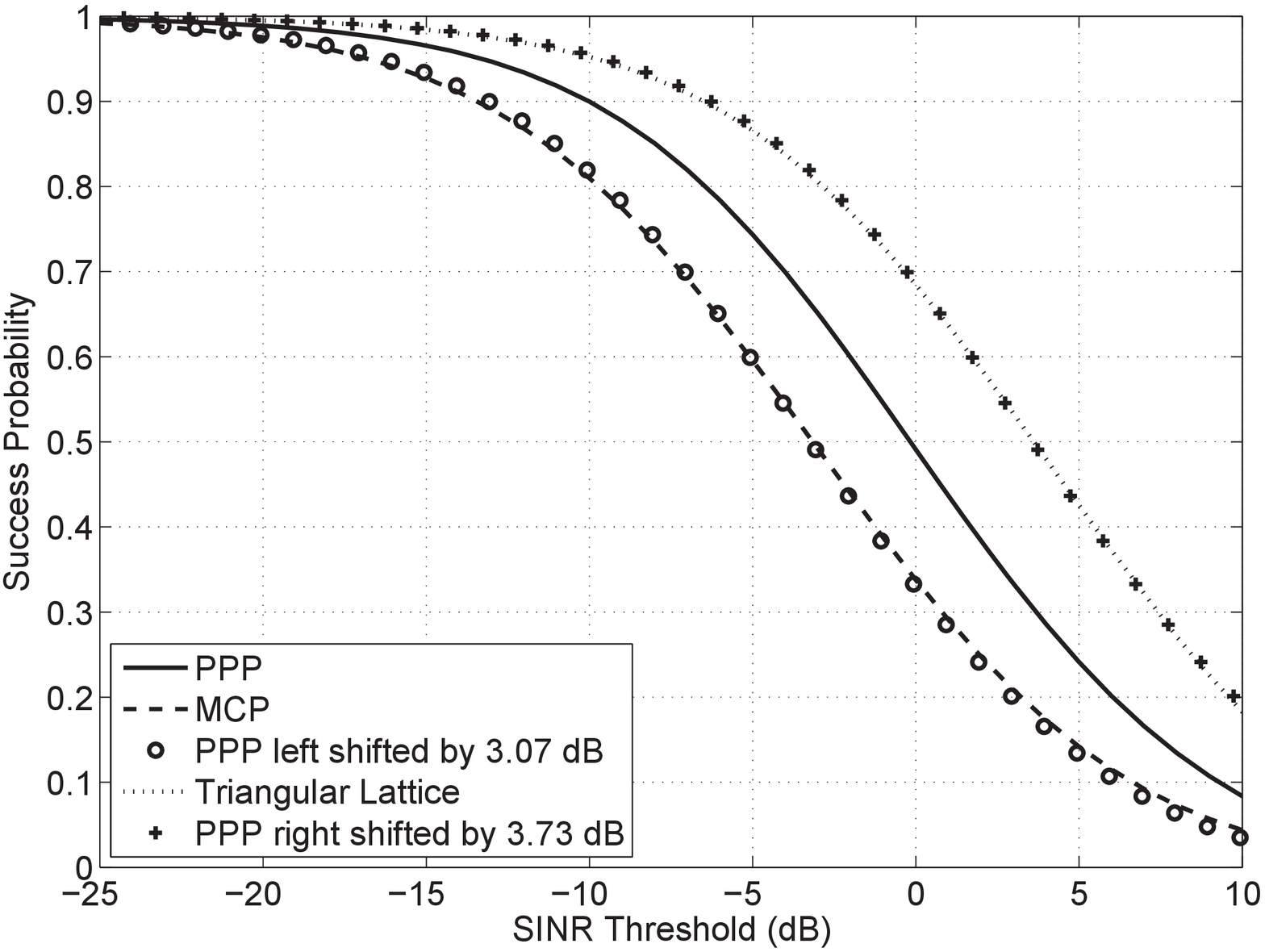}\\
  \caption{The  success probability of the PPP with intensity $\lambda = 0.1$, the MCP with $\lambdap = 0.01$, $\bar{c} = 10$ and $\rc = 5$ (see Section \ref{sec:SC-PP} for an explanation of these  parameters),      
  and the triangular lattice with density $\lambda = 0.1$   for Rayleigh fading, path loss model $\ell(x)= (1+\|x\|^{4})^{-1}$ and no noise, which are simulated on a $100 \times 100$ square.
The  lines are the success  probability curves of the three point processes, while the markers indicate  the success  probability curves of the PPP shifted by the DGs of the MCP and the triangular lattice at $\pt = 0.6$.  }
  \label{f0}
\end{figure}
\else
\begin{figure}[hbtp]
\centering
  \includegraphics[width=\figwidth]{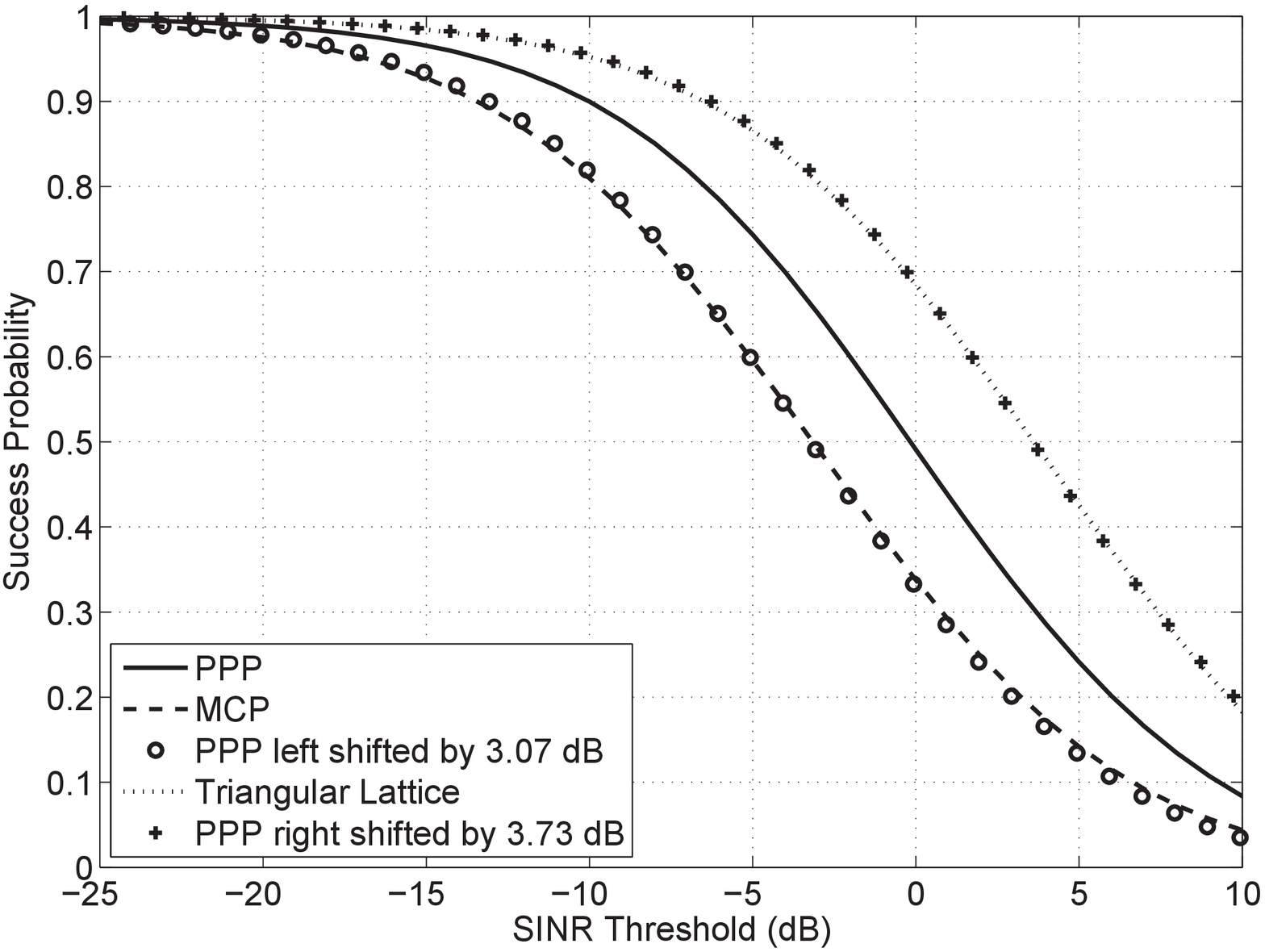}\\
  \caption{The  success probability of the PPP with intensity $\lambda = 0.1$, the MCP with $\lambdap = 0.01$, $\bar{c} = 10$ and $\rc = 5$ (see Section \ref{sec:SC-PP} for an explanation of these  parameters),      
  and the triangular lattice with density $\lambda = 0.1$   for Rayleigh fading, path loss model $\ell(x)= (1+\|x\|^{4})^{-1}$ and no noise, which are simulated on a $100 \times 100$ square.
The  lines are the success  probability curves of the three point processes, while the markers indicate  the success  probability curves of the PPP shifted by the DGs of the MCP and the triangular lattice at $\pt = 0.6$.  }
  \label{f0}
\end{figure}
\fi

Note that, the ADG may not exist for some point processes and fading types. In the following section, we  will provide some sufficient conditions for the existence of the ADG.  
For Rayleigh fading, the ADG of the MCP exists. 

Similar to the DG, the ADG  measures the success probability but characterizes the difference between the success  probability of the PPP and a given point process as the success probability approaches 1 instead of for a target success probability, and by observation from Fig.~\ref{f0}, the ADG closely approximates the DG for all practical values of the success probability.     
Hence, given the ADG of a point process, we can evaluate its success probability by shifting (in dB) the corresponding PPP results, that is to say, $\Pc(\theta) \approx \Pcppp(\theta/\hat{G})$ and $\Pc(\theta) \sim \Pcppp(\theta/\hat{G})$, $\theta \to 0$. In Fig. \ref{f0}, we observe that $\hat{G} \approx 2.4$  for the triangular lattice and $\hat{G} \approx 0.5$  for the MCP.
Note that the ADG relative to the PPP permits an immediate calculation of the ADG between two arbitrary point processes.

\section{Existence of the Asymptotic Deployment Gain}
In this section, we derive several  important asymptotic properties of the SINR distribution, given some general assumptions about  the point process  and  the CDF of the fading variables. 
These asymptotic properties, in turn, prove the existence of the ADG. 

\subsection{Definition of a General Class of Base Station Models}
First we give several notations, based on which we introduce the precise class of point processes we focus on.
We define  the contact distance $\xi \triangleq \| \np_{\Phi}(o) \|$, and
define the supremum of $\xi$ as 
$\xi_{\textrm{max}} \triangleq \sup_{x\in \mathbb{R}^2} \min_{y\in \Phi}\{ \|x-y\| \}$.  
Due to the ergodicity of the point process (which follows
from the mixing property) \cite[Ch.~2]{book}, $\xi_{\max}$ does not depend on the realization of $\Phi$.
 $\xi_{\textrm{max}} = \infty$
in many  mixing point processes. 

We define $\Phi_o^\zeta\triangleq(\Phi\mid\np_{\Phi}(o)=\zeta)$, where $\zeta \in \mathbb{R}^2 \setminus \{o\}$,   
as the conditional point process that satisfies $\np_{\Phi}(o)=\zeta$, which
implies  $\zeta\in\Phi_o^\zeta$ and $\Phi_o^\zeta(b(o,\|\zeta\|))=0$.\footnote{For a point process $\Phi$, $\Phi(B)$ is a random variable that denotes the number of points in set $B \subset \mathbb{R}^2$.}
So given that $\zeta$ is the closest point of $\Phi$ to $o$, the total interference is $\Iy$. 
However, it is tricky to directly handle  the conditional point process conditioned on that there is an empty disk,  if the point process is not the PPP. 
Thus, we compare the interference in $\Phi_o^\zeta$ with the interference from a point process where the desired BS $\zeta$ is not necessarily the closest one. To this end, we define $\Phi^\zeta \triangleq (\Phi\mid\zeta\in\Phi)$ and consider its interference outside  a disk of radius $\|\zeta\|/2$ around the origin:
\begin{equation}
\IzII = \sum_{x\in \Phi^{\zeta} \bigcap B_{\zeta/2} \setminus \{\zeta\}} h_x\ell(x),
\end{equation}
where $B_{\zeta/2} \triangleq  \mathbb{R}^2 \setminus b(o,\|\zeta\|/2)$.  
Note that it is not necessary to set the radius of the disk to $\|\zeta\|/2$; in fact,  the radius could be any quantity that is smaller than $\|\zeta\|$. 
Since we can use standard Palm theory \cite{book} for $\Phi^\zeta$,  it is easier to deal with $\Phi^\zeta$ than $\Phi_o^\zeta$.   

To motivate the above notations, we give an  illustration of them in Fig. \ref{illu1}. Both $\Phi^{\zeta}_o$ and $\Phi^{\zeta}$ have a point at $\zeta$, and we let $\|\zeta\| = y$. All points of $\Phi^{\zeta}_o$ are  located in the striped region  (outside $b(o,y)$), and $\Iy$ is the interference from all these points except $\zeta$.  In contrast,    $\Phi^{\zeta}$ may
have points throughout the whole plane, but  $\IzII$ is the interference only  from the points of $\Phi^{\zeta}$ in the shaded region (outside $b(o,y/2)$) except $\zeta$.

\solution
\begin{figure}[hbtp]
\centering
  \includegraphics[width=0.5 \columnwidth]{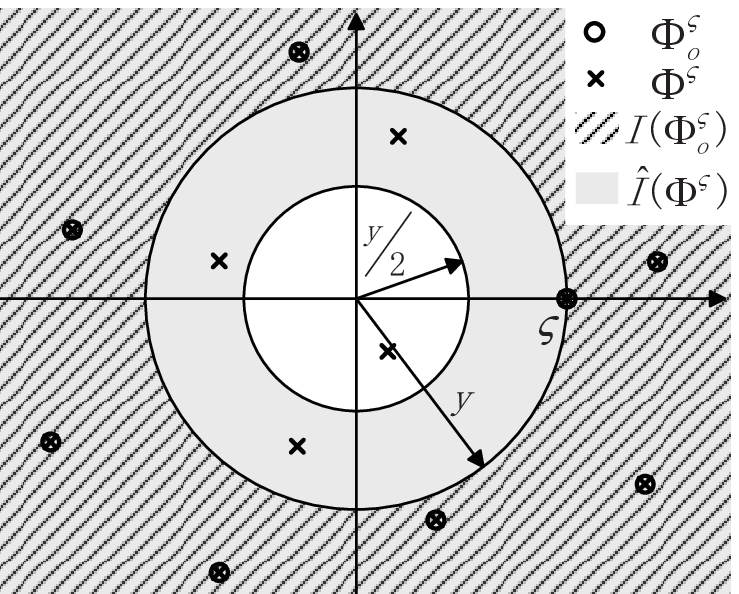}\\    
  \caption{An illustration of  $\Phi^{\zeta}_o$, $\Phi^{\zeta}$, $\Iy$ and $\IzII$, where $ \|\zeta\| = y$.} 
  \label{illu1}
\end{figure}
\else
\begin{figure}[hbtp]
\centering
  \includegraphics[width=0.8\figwidth]{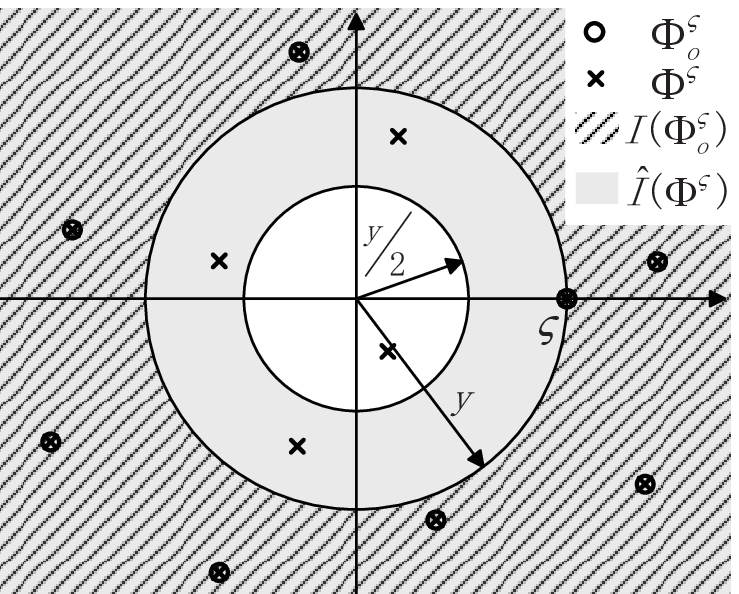}\\    
  \caption{An illustration of  $\Phi^{\zeta}_o$, $\Phi^{\zeta}$, $\Iy$ and $\IzII$, where $ \|\zeta\| = y$.} 
  \label{illu1}
\end{figure}
\fi

Using the above notations,  we define a general class of point process distributions that we use to rigorously state our main result on the SINR distribution. 
\begin{definition}[Set $\mathcal{A}$]
The set $\mathcal{A}=\{P_{\Phi}\}$ is the set of all m.i. point process distributions $P_{\Phi}$ that are mixing and that satisfy the following four conditions. If a point process $\Phi$ is distributed as $P_{\Phi}\in \mathcal{A}$,
\begin{enumerate}
\item for all $n \geq 2$, the $n$th moment density of $\Phi$ is bounded,  
i.e., 
$\exists q_n<\infty$, such that
$\rho^{(n)}(x_1,\ldots,x_n) < q_n$, 
for  $x_1,\ldots,x_n   \in \mathbb{R}^2$;
\item for all $y>0$, $\exists \zeta \in \mathbb{R}^2$ with $\|\zeta\| = y$, such that   
$\mathbb{P}(\Phi^{\zeta}(b(o,y))=0 ) \neq 0$;   
\item $\exists y_0 > 0$, such that for all $y>y_0$ and $\zeta \in \mathbb{R}^2$ with $\|\zeta\| = y$,     
$\IzII$ stochastically dominates $\Iy$, i.e., $\mathbb{P}(\Iy>z) \leq \mathbb{P}(\IzII>z)$,  for all $z \geq 0$;
\item $\forall n \in \mathbb{N}$, the $n$-th moment of   the contact distance  $\xi$ is bounded, i.e., $\exists b_n < \infty$, s.t. $\mathbb{E}(\xi^n) <b_n$.
\end{enumerate}
\label{setA}
\end{definition}

To clarify the need for these conditions, we offer the following remarks: 
\begin{enumerate}
\item  By the mixing property, we have that $\rho^{(2)}(x_1,x_2) \to \lambda^2$, as $\|x_1-x_2\| \to \infty$, which indicates that when $\|x_1-x_2\|$  is large enough, $\rho^{(2)}(x_1,x_2)$ is bounded. The first condition is stronger than that. It guarantees that the $n$th moment measure of $\Phi$ is absolutely continuous with respect to the Lebesgue measure, which, in turn, implies that $\Phi$ is \emph{locally finite} \cite[Ch.~2.2]{book}. A point process is locally finite if and only if $\Phi(B) < \infty$ a.s., for any $B \subset \mathbb{R}^2$ with $\nu(B)<\infty$, where $\nu(\cdot)$ is the Lebesgue measure. Local  finiteness is  a standard assumption in point process theory, but it is too weak for our purposes.  
For example, Condition 1  excludes some extreme{\footnote{We call this point process extreme since, conditioned on a point at $o$, there is a positive probability of having another point on a subset of Lebesgue measure zero.}} cases, such as the Gauss-Poisson point process as described in \cite[Sec.~3.4]{book}, which is locally finite. 
\item Since $\Phi$ is a m.i.~point process, the second condition  is equivalent to requiring  that for all $y>0$, $\forall \zeta \in \mathbb{R}^2$ with $\|\zeta\| = y$, such that   
$\mathbb{P}(\Phi^{\zeta}(b(o,y))=0 ) \neq 0$. That is to say, if $\zeta \in \Phi$,  
the probability of no points of $\Phi$ being located in $b(o,\|\zeta\|)$ is  positive.  
The condition also 
 implies that $\xi_{\rm max } =  \infty$. Because if $\xi_{\rm max } < \infty$,  for all $y > \xi_{\rm max } $, 
 there surely is at least one point of $\Phi$ in $b(o, y)$, which leads to a contradiction since it would imply that  
 $\mathbb{P}(\Phi^{\zeta}(b(o,y))  > 0 ) = 1$.
So,  the condition  excludes the m.i.~and mixing point processes where there exists $r_0>0$, such that for all $x \in \mathbb{R}^2$, there is at least one point in the region $b(x,r_0)$.  Those point processes\footnote{Note that the perturbed triangular lattice\cite{AGTWC13} is not in the exclusion category, since it is not mixing.} may be constructed, but are rarely considered in the context of  wireless networks. 
\item The third condition is based on  the two random variables $\Iy$ and $\IzII$, whose expressions contain the fading variables. But, in fact, the condition is independent of the fading type, since the fading variables are i.i.d.~and their expectation is bounded.  
The condition means that  there exists $y_0 > 0$, such that for all $y>y_0$ and  $\zeta \in \mathbb{R}^2$ with $\|\zeta\| = y$, the CCDF of the interference from $\Phi^{\zeta} \bigcap B_{\zeta/2} \setminus \{\zeta\}$ is always no smaller than  the CCDF of the interference from $\Phi_o^\zeta \setminus \{\zeta\}$. 
Most point processes meet the condition, since an extra region $b(o,y) \setminus b(o,y/2)$ is included in $\IzII$, but not in $\Iy$. Some point processes, which are seldom considered,  may violate the condition. For example, albeit somewhat artificial, for small $\epsilon >0$, the expectation of $\Phi_o^{\zeta}(b(o,\|\zeta\|+\epsilon))$ is much greater than that of $\Phi^{\zeta}(b(o,\|\zeta\|+\epsilon))$, which, at last, leads to the violence of the third condition.   
Such kind of point processes are beyond our consideration.
\item The fourth condition is satisfied by most point processes that are considered. A sufficient condition of the fourth condition is that $F^{\rm c}_{\xi}(x) < \exp(-c_0x)$, as $x \to \infty$, where $F^{\rm c}_{\xi}$ is the CCDF of $\xi$ and $c_0 \in \mathbb{R}^{+}$. One simple example is the PPP with intensity $\lambda$, whose CCDF of $\xi$ is $F^{\rm c}_{\xi}(x) = \exp(-\lambda \pi x^2)$. 
\end{enumerate}

In summary, the four conditions in Def.~\ref{setA} are quite  mild;  they   are satisfied by
most point processes that are usually considered in wireless networks and in stochastic geometry,
such as the PPP, the MCP,  the Mat\'{e}rn hard-core process (MHP) \cite[Ch.~3]{book} and the Ginibre process\cite{NaTWC14}.
The triangular lattice is not included, since it is not  mixing and $\xi_{\textrm{max}} <  \infty$.
We will prove that the laws of the PPP, the MCP and the MHP belong to $\mathcal{A}$ in Section \ref{sec:SC-PP}.

\subsection{Main Results}

Before  presenting  the main theorem, we state a property of the distribution of  $\Iy$. 
\begin{lemma}
Assume the fading variable $h$ satisfies that $\forall n\in \mathbb{N}$, $\mathbb{E}(h^n)<+\infty$. 
For a point process $\Phi$ with $P_{\Phi}\in \mathcal{A}$, the following statements hold:
\begin{enumerate}
\item for $\ell(x) = ( 1+\| x\|^{\alpha})^{-1}$, all moments of the interference $\Iy$ are  bounded, i.e., $\forall n \in \mathbb{N}$,  $\exists c_n \in \mathbb{R}^+$, such that  $\mathbb{E}(\Iy^n)<c_n$, where $c_n$ does not depend on $\zeta$; 

\item for $\ell(x) =  \| x\|^{-\alpha}$, all moments of the interference $\Iy$ are  bounded, and $\forall n \in \mathbb{N}$, $\exists c_n \in \mathbb{R}^+$, such that $\mathbb{E}(\Iy^n)<c_n \max\{1,\|\zeta\|^{2-\alpha n}\} $.
\end{enumerate}
\label{lemmaIRDirect}
\end{lemma}
\begin{proof}
See Appendix \ref{sec:appE}.
\end{proof}

Since $\Iy$ can be interpreted as the total interference at $o$ if the nearest base station to $o$ is at $\zeta$, 
Lemma \ref{lemmaIRDirect} shows that all  moments of the total interference are bounded. If the path loss model is non-singular, the bound can be chosen to be independent of $\|\zeta\|$.  
However, if the path loss model is singular, the bound depends on $\|\zeta\|$, and if $\|\zeta\|$ goes to 0, it can be proved that $\mathbb{E}(\Iy)$  becomes arbitrarily large  for some BS processes, e.g., the PPP.\footnote{Note that for the PPP with intensity $\lambda$, if we de-conditioned on $\|\zeta\|$, by Campbell's theorem, the mean interference $\mathbb{E}(I(\Phi)) = \sum_{x \in \Phi \setminus  \np_{\Phi}(o)} h_x \|x\|^{-\alpha} = \mathbb{E}(h) \int_0^{\infty} \frac{(2 \pi \lambda)^2}{\alpha-2}  x^{3-\alpha} e^{-\lambda \pi x^2}dx$. So, $\mathbb{E}(I(\Phi))$ is finite for $2<\alpha < 4$,  and infinite for $\alpha \geq 4$.} 

Now we are equipped to state our main result: 
if the CDF of the fading variable $h$ decays polynomially around 0 and all moments of $h$ are bounded, then as a result of the  boundedness of the moments of the interference, the outage probability $1-\Pc(\theta)$ expressed in dB, as a function of the SINR threshold $\theta$, also in dB, has the same slope as $\theta \to 0$, for all $\Phi$ with $P_{\Phi}\in \mathcal{A}$.

\begin{theorem}
For  a point process $\Phi$ with $P_{\Phi}\in \mathcal{A}$, if the fading variable satisfies 
\begin{enumerate}
\item  $\exists m \in (0, +\infty)$, s.t. $F_h(t) \sim a t^m$, as $t \to 0$, 
where $a>0$ is constant, 
\item  $\forall n\in \mathbb{N}$, $\mathbb{E}(h^n)<+\infty$, 
\end{enumerate}
then we have
\begin{equation}
 \frac{1-\Pc(\theta)}{\theta^m} \to \kappa,  \quad \textrm{ as } \theta \to 0,
 \label{DGsigma}
\end{equation}
where $0<\kappa < \infty$  does not depend on $\theta$ and is  given by   
\begin{align}
\kappa = & \int_0^{\infty} \mathbb{E}_{\Iy} \Big[ a  \ell(y) ^{-m} \big(\Iy   + W \big)^m\Big]   f_{\xi}(y) dy
 \label{deltaexp}
\end{align}
$ ( \|\zeta\| = y ) $  and $f_{\xi}$ is the PDF of $\xi$. 

\label{thmain}
\end{theorem}

\begin{proof}
See Appendix \ref{proof:thmain}.
\end{proof}

Theorem 1 shows that the ADG exists and     
how it depends on the  other network parameters. The following theorem quantifies the ADG. 
\begin{corollary}
Under the same condition as in Theorem \ref{thmain}, 
the ADG of $\Phi$ exists and is given by
\begin{equation}
\hat{G} =   \Big( \frac{\kappa^{\rm PPP} }{ \kappa} \Big)^{\frac{1}{m}},
\end{equation}
where $\kappa^{\rm PPP}$ is the value for the PPP and $\kappa$  is the value for $\Phi$. For the PPP with intensity $\lambda$,
\solution
\begin{equation}
\kappa^{\rm PPP} =  2\lambda \pi \int_0^{\infty} \mathbb{E}_{I_r} \Big[ \frac{m^{m-1}}{\Gamma(m)}  \ell(r) ^{-m} \big(I_r   + W \big)^m\Big]  r \exp(-\lambda \pi r^2) dr,
\end{equation}
\else
\begin{align}
\kappa^{\rm PPP} &=  2\lambda \pi \int_0^{\infty} \mathbb{E}_{I_r} \Big[ \frac{m^{m-1}}{\Gamma(m)}  \ell(r) ^{-m} \big(I_r   + W \big)^m\Big]    \nonumber \\
&\quad \cdot r \exp(-\lambda \pi r^2) dr,
\end{align}
\fi
where $I_r = \sum_{x\in \Phi  \bigcap b(o,r)^{\rm c}} h_x\ell(x)$.
\label{Corol2}
\end{corollary}
\begin{proof}
Given a target success probability $\pt$, define $\theta_1 \triangleq \Pc^{-1}(\pt)$ and $\theta_2  \triangleq (\Pcppp)^{-1}(\pt)$. As $\pt \to 1$, we have $\theta_1 \to 0$ and $\theta_2 \to 0$. By Theorem \ref{thmain}, 
$1-\Pc(\theta_1) \sim \kappa \theta_1^m$ and $1-\Pcppp(\theta_2) \sim \kappa^{\rm PPP} \theta_2^m$. 
Since $\pt = \Pc(\theta_1) = \Pcppp(\theta_2)$, as $\pt \to 1$, $  \kappa \theta_1^m = \kappa^{\rm PPP} \theta_2^m$.
Thus, $\hat{G} = \lim_{\pt \to 1} \theta_1/\theta_2 =   (  {\kappa^{\rm PPP} }/{ \kappa} )^{{1}/{m}}$.
\end{proof}

Note that Rayleigh fading meets the requirements in Theorem \ref{thmain} with $m=1$. For the special case of the PPP with intensity $\lambda$, no noise and Rayleigh fading, it has been shown in \cite{Tract} that 
\begin{equation}
\Pc(\theta) =\bigg({1+{\theta}^{\delta} \int_{\theta^{-\delta}}^{\infty} \frac{1}{1+u^{1/\delta}}du}\bigg)^{-1}, 
\label{Success_PPP}
\end{equation} 
where $\delta \triangleq 2/\alpha$.  It follows that $\kappa^{\rm PPP} = \lim_{\theta \to 0} \frac{1-\Pc(\theta)}{\theta} = \frac{2}{\alpha-2}$.  For $\alpha = 4$, $\Pc(\theta) = {1}/(1+\sqrt{\theta} \arctan  {\sqrt{\theta}}  )$, and $\kappa^{\rm PPP} = 1$. 

A point process has different  ADGs  depending on  the value of  $m$. So it is sensible to  compare the ADGs of  different point process models only under the same fading assumption.

We have proved that the ADG exists  with certain constrains on the fading and point processes. In the rest of this section, we consider some special cases.  

\subsection{Special Cases - Fading Types}

Regarding the fading, we mainly consider Nakagami-$m$ fading and  \emph{composite fading}, which is a combination of  Nakagami-$m$ fading and log-normal shadowing. 

\subsubsection{Nakagami-$m$ Fading}
The fading variable $h \sim \textrm{gamma}(m,\frac{1}{m})$. On the one hand,   we have 
\begin{equation}
 \lim_{t \to 0} \frac{F_h(t)}{t^m}  = \lim_{t \to 0} \frac{(mt)^{m-1}\exp(-mt)}{\Gamma(m) t^{m-1}}=\frac{m^{m-1}}{\Gamma(m)} < +\infty.
\end{equation}
On the other hand, since $F_{h}^{\rm c}(x)$ has an exponential tail, all moments of $h$ are finite. Thus, Nakagami-$m$ fading meets the requirements in Theorem \ref{thmain}. 

In addition, we find an interesting phenomenon that for a point process $\Phi$ with  $P_{\Phi}\in \mathcal{A}$,  the behavior of the CCDF of the fading at the tail determines the tail behavior of the CCDF of the interference $\Iy$.  
The following corollary formalizes this    property. 
As usual, $f(x)=\Omega(g(x))$ as $x\to \infty$   means $\limsup_{x\to \infty}\big|\frac{f(x)}{g(x)} \big| > 0$.

\begin{corollary}
For   a point process $\Phi$ with  $P_{\Phi}\in \mathcal{A}$,   
if  the fading has at most an exponential tail, i.e.,  $-\log F_{h}^{\rm c}(x)  = \Omega(x)$,  
$x \to \infty$, where $F_{h}^{\rm c}(x)$ is the CCDF of the fading variable $h$,
then    
the interference tail is bounded by an exponential, i.e.,  $-\log F_{\Iy}^{\rm c}(x)  = \Omega(x)$, $x \to \infty$,  where $F_{\Iy}^{\rm c}(x)$ is the CCDF of $\Iy$.
\label{lemmaIR}
\end{corollary}
\begin{proof}
See Appendix \ref{proof:distrOfI}.
\end{proof}

A similar property has been derived in \cite{MIND},
 namely, that in ad hoc networks modeled by m.i. point processes, an exponential tail in the fading distribution implies an exponential tail in the interference distribution.
The result cannot be directly applied to cellular networks, because  in the cellular network that  we consider,   each user communicates with its nearest BS $u$  and thus no interferers can be closer than $u$, while
the authors in \cite{MIND} assume the receiver communicates with a transmitter at a fixed location and there can be some interferers closer to the receiver than the transmitter.

\subsubsection{Composite Fading}
The signals from all BSs experience both Nakagami-$m$ fading and log-normal shadowing.
A similar kind of fading has been investigated in \cite{CompositeFading1, CompositeFading2}, where the fading was composed of Rayleigh fading and   log-normal shadowing.   
Denoting the fading variable with respect to Nakagami-$m$ fading by $\tilde{h}$ and the  fading variable with respect to  log-normal shadowing by $\hat{h}$, the composite fading variable can be represented as $h = \tilde{h} \hat{h}$, where $\tilde{h}$ and $\hat{h}$ are independent.

For  log-normal shadowing, we use the definition from \cite{logNormaldef}. 
Without  loss of generality, we assume  $\hat{h} =10^{X/10}$, where $X \sim N(0,\sigma^2)$. 
The CDF of $\hat{h} $, denoted by $F_{\hat{h}}(t)$, is
\solution
\begin{align}
F_{\hat{h}}(t) = \frac{1}{2}{\rm erfc}\bigg(- \frac{10\log t}{\sigma \sqrt{2} \log{10}} \bigg) = \frac{1}{\sqrt{\pi}} \int_{- \frac{10\log t}{\sigma \sqrt{2} \log10}}^{\infty} \exp(-v^2)dv,   
\end{align}
\else
\begin{align}
F_{\hat{h}}(t) &= \frac{1}{2}{\rm erfc}\bigg(- \frac{10\log t}{\sigma \sqrt{2} \log{10}} \bigg)   \nonumber \\
&= \frac{1}{\sqrt{\pi}} \int_{- \frac{10\log t}{\sigma \sqrt{2} \log10}}^{\infty} \exp(-v^2)dv,   
\end{align}
\fi
where ${\rm erfc}$ is the complementary error function. 
It is straightforward to obtain that\footnote{Note that the mean of $\hat{h}$ is not 1. Actually, we could normalize it to 1 and replace it with the normalized variable in our results, but since it does not affect our results, for convenience, we just leave it as it is. }  $\mathbb{E}[{\hat{h}}] = \exp((\frac{\log 10}{10})^2 \frac{\sigma^2}{2})$ and $\mathbb{E}[{\hat{h}}^2] = \exp((\frac{\log 10}{10})^2 2\sigma^2 )$, and to   show that as $t \to \infty$, $F_{\hat{h}}^{\rm c}(t) $ decays faster than $t^{-n}$ for any $n \in \mathbb{N}$, but slower than $\exp(-at)$ for any $a>0$.

For  composite fading, we have the following  lemma about the distribution of $h$.

\begin{lemma}
If $\tilde{h} \sim  \textrm{gamma}(m,\frac{1}{m})$, $10\log \hat{h}/\log 10 \sim N(0, \sigma^2)$, and $\tilde{h}$ is independent of $\hat{h}$, the distribution of $h=\tilde{h} \hat{h}$ has the following properties:
\begin{enumerate}
\item  $F_h$ decays polynomially around 0 and 
\solution
\begin{equation}
\lim_{t \to 0} \frac{F_h(t)}{t^m} =  \int_0^{\infty}  \frac{10 m^{m-1}}{\sigma \log10\sqrt{2\pi} \Gamma(m) u^{m+1}}  \exp \bigg(-\Big(\frac{10\log u}{\sigma \sqrt{2} \log10}\Big)^2 \bigg) du < \infty;    \nonumber
\end{equation}
\else
\begin{align}
\lim_{t \to 0} \frac{F_h(t)}{t^m} &=  \int_0^{\infty}  \frac{10 m^{m-1}}{\sigma \log10\sqrt{2\pi} \Gamma(m) u^{m+1}}  \nonumber \\
&\quad \cdot \exp \bigg(-\Big(\frac{10\log u}{\sigma \sqrt{2} \log10}\Big)^2 \bigg) du < \infty;  \nonumber
\end{align}
\fi
\item $F_h^{\rm c}(t) = o(t^{-n})$, as $t \to \infty$, for any $n \in \mathbb{N}$, and $-\log F_h^{\rm c}(t)  = o(t)$, $t \to \infty$.
\end{enumerate}
\label{LemmaCompound}
\end{lemma}
\begin{proof}
See Appendix \ref{sec:appF}.
\end{proof}

The two properties in Lemma \ref{LemmaCompound} indicate that the composite fading retains the asymptotic property of Nakagami-$m$ fading for $t \to 0$ and that of log-normal shadowing for $t \to \infty$, respectively.  
They also imply that the composite fading meets the requirements in Theorem \ref{thmain}. 

Regarding  the distribution of the interference at the tail, we  have the following corollary.

\begin{corollary}
For  a point process $\Phi$ with $P_{\Phi}\in \mathcal{A}$ and  composite fading, the interference tail is upper bounded by a power law with arbitrary parameter $\beta$, i.e., $F^{\rm c}_{\Iy} ( y) = o(y^{-\beta})$, $\forall \beta \in \mathbb{N}$, as $y \to +\infty$. 
\label{lemmaIRSHAWDOW}
\end{corollary}
\begin{proof}
We can simply apply the Markov inequality and have that $\forall \beta \in \mathbb{N}$,
\begin{align}
 \mathbb{P}(\Iy> y )&\leq  \frac{\mathbb{E}(\Iy^{\beta})}{y^{\beta}}.
\end{align}

Hence, using  Lemma \ref{lemmaIRDirect}, we have $F^{\rm c}_{\Iy} ( y) = o(y^{-\beta})$, $\forall \beta \in \mathbb{N}$, as $y \to +\infty$.
\end{proof}

\subsection{Special Cases - Point Processes}
\label{sec:SC-PP}
As for the point processes, we  specifically concentrate on the PPP, the MCP and the MHP.
We first briefly describe the MCP and the MHP.

\emph{Mat\'{e}rn Cluster Process}: As a class of clustered point processes on the plane built on a PPP, the MCPs  
 are  doubly Poisson cluster processes, where the parent points  form a uniform PPP $\Phi_{\rm p}$ of intensity $\lambdap$ and the daughter points are uniformly scattered on the ball of radius $\rc$ centered at each parent point $x_{\rm p}$ with intensity 
$\lambda_0 (x) = \frac{\bar{c}}{\pi \rc^2} \mathbf{1}_{B(x_{\rm p},\rc)}(x)$, 
 where $B(x_{\rm p},\rc) \triangleq \{ x\in \mathbb{R}^2: \| x - x_{\rm p} \| \leq  \rc \}$ is the closed disk of radius $ \rc$ centered at $x_{\rm p}$.
The mean number of daughter points in one cluster is $\bar{c}$. So the intensity of the process is $\lambda = \lambdap \bar{c}$.

\emph{Mat\'{e}rn Hard-core Process}: The MHPs are a class of repulsive point processes, where points are forbidden to be closer than a certain minimum distance.  There are several types of  MHPs. Here we only consider the MHP of type II \cite[Ch.~3]{book}, which is generated by starting with a basic uniform PPP $\Phi_{\rm b}$ of intensity $\lambdab$, adding to each point $x$ an independent random variable $m(x)$, called a mark, uniformly distributed on $[0,1]$, then flagging for removal all points that have  a neighbor within distance $\rh$ that has a smaller mark and finally  removing all flagged points.  
The intensity of the MHP is $\lambda = \frac{1-\exp(-\lambdab \pi \rh^2)}{\pi \rh^2}$. The highest density $\lambda_{{\rm max}} = 1/(\pi \rh^2)$ is achieved as $\lambdab \to \infty$.

\begin{lemma}
The distributions of the PPP, the MCP and the MHP belong to the set $\mathcal{A}$.  
\label{lemma3PP}
\end{lemma}
\begin{proof}
See Appendix \ref{sec:appC}.
\end{proof}

By Lemma \ref{lemma3PP}, regarding Nakagami-$m$ fading and composite fading, we have the following corollary directly from Theorem \ref{thmain}. 

\begin{corollary}
If the fading is  Nakagami-$m$ or the composite fading,  then for the PPP, the MCP and the MHP,
\begin{equation}
 \frac{1-\Pc(\theta)}{\theta^m} \to {\kappa}, \quad \textrm{ as } \theta \to 0,
\end{equation}
where $\kappa$   is  given by \eqref{deltaexp}.   
In \eqref{deltaexp}, for Nakagami-$m$ fading, $a =  \frac{m^{m-1}}{\Gamma(m)} $; for the composite fading, $a =  \int_0^{\infty}  \frac{10 m^{m-1}}{\sigma \log10\sqrt{2\pi} \Gamma(m) u^{m+1}}  \exp(-(\frac{10\log u}{\sigma \sqrt{2} \log10})^2) du$.
\label{specialcase}
\end{corollary}

\section{Applications of the Asymptotic Deployment Gain}
Since the ADG characterizes the gap of the success probability between a point process and the PPP, any statistic that depends on the distribution of the SINR (e.g., the average ergodic rate and the mean SINR) can be approximated using the ADG. 
In this section, we focus on the average ergodic rate  and the mean SINR. 

\subsection{Average Ergodic Rate }
We assume base station adopts adaptive modulation/coding to achieve the Shannon bound of the rate for the instantaneous SINR. That is to say, each BS adjusts its rate of transmission to $\gamma = \ln(1+ \textrm{SINR})$.  
The average ergodic rate (expressed in nats) is $\bar{\gamma} \triangleq \mathbb{E}[\ln(1+\textrm{SINR})]$.

Denoting  the ADG of $\Phi$ as $\hat{G}$ and the success probability of the corresponding PPP as $\Pc^{\rm PPP}(\theta)$, the success probability for $\Phi$ is approximated as $\Pc^{\rm PPP}(\theta/\hat{G})$.  
The average ergodic rate can be expressed as 
\solution
\begin{align}
\bar{\gamma} &\approx -\int_{0}^{\infty} \ln(1+\theta) d\Pc^{\rm PPP}\bigg(\frac{\theta}{\hat{G}}\bigg) =  -\int_{0}^{\infty} \ln(1+\hat{G} \theta) d\Pc^{\rm PPP}(\theta)   \stackrel{(a)}{=} \int_{0}^{\infty} \Pc^{\rm PPP}\bigg( \frac{e^{x} - 1}{\hat{G}}\bigg)dx,  \nonumber
\end{align}
\else
\begin{align}
\bar{\gamma} &\approx -\int_{0}^{\infty} \ln(1+\theta) d\Pc^{\rm PPP}\bigg(\frac{\theta}{\hat{G}}\bigg) \nonumber \\
&=  -\int_{0}^{\infty} \ln(1+\hat{G} \theta) d\Pc^{\rm PPP}(\theta)   \nonumber \label{cdf_1}\\
& \stackrel{(a)}{=} \int_{0}^{\infty} \Pc^{\rm PPP}\bigg( \frac{e^{x} - 1}{\hat{G}}\bigg)dx, \nonumber
\end{align}
\fi
where  $(a)$ follows since the CCDF of the random variable $X = \ln(1+\hat{G} \cdot \textrm{SINR})$ is 
$\mathbb{P}(X>x) = \mathbb{P}\big(\textrm{SINR} > (e^{x} - 1)/\hat{G} \big)  = \Pc^{\rm PPP}\big((e^{x} - 1)/\hat{G} \big)$ and the expectation of a positive  random variable can be expressed as the integral over the CCDF.  

\subsection{Mean SINR}
Just as the success probability and the average ergodic rate, the mean SINR is also an important criterion that has been discussed in wireless networks, e.g. in \cite{MeanSINR}.  Denote $M_{\Phi}$ as the mean SINR for $\Phi$, and $M_{\rm PPP}$ the mean SINR for the PPP with the same intensity as that of $\Phi$. 
It can be proved that the mean SINR for the PPP is infinite if the path loss model is singular.  
Briefly, for  $\zeta = \np_{\Phi}(o)$, letting $y = \|\zeta\|$, we have 
\solution
\begin{align}
\mathbb{E}(\textrm{SINR}) &= \mathbb{E}\bigg(\frac{\ell(\zeta)}{ W + I(\Phi_o^{\zeta})}\bigg) \stackrel{(a)}{\geq}  \mathbb{E}_{y}\bigg( \frac{\ell(\zeta)}{ W + \mathbb{E}[I(\Phi_o^{\zeta})]} \bigg) \stackrel{(b)}{\geq}    \mathbb{E}_{y}\bigg( \frac{y^{-\alpha}}{ W + c_1 \max\{1,y^{2-\alpha}\}} \bigg) \nonumber \\
&=  \int_0^1    \frac{x^{-\alpha}}{ W + c_1  x^{2-\alpha}} f_{\|\zeta\|}(x) dx +  \int_1^{\infty}    \frac{x^{-\alpha}}{ W + c_1 } f_{\|\zeta\|}(x) dx   \nonumber \\
&\geq    \int_0^1    \frac{x^{-1}}{ W + c_1 } 2\pi \lambda e^{-\lambda \pi x^2} dx +  \int_1^{\infty}    \frac{x^{-\alpha}}{ W + c_1 } f_{\|\zeta\|}(x) dx  = \infty,  \nonumber
\end{align}
\else
\begin{eqnarray}
&&\mathbb{E}(\textrm{SINR}) = \mathbb{E}\bigg(\frac{\ell(\zeta)}{ W + I(\Phi_o^{\zeta})}\bigg) \stackrel{(a)}{\geq}  \mathbb{E}_{y}\bigg( \frac{\ell(\zeta)}{ W + \mathbb{E}[I(\Phi_o^{\zeta})]} \bigg) \nonumber \\
&&\stackrel{(b)}{\geq}    \mathbb{E}_{y}\bigg( \frac{y^{-\alpha}}{ W + c_1 \max\{1,y^{2-\alpha}\}} \bigg) \nonumber \\
&&=  \int_0^1    \frac{x^{-\alpha}}{ W + c_1  x^{2-\alpha}} f_{\|\zeta\|}(x) dx +  \int_1^{\infty}    \frac{x^{-\alpha}}{ W + c_1 } f_{\|\zeta\|}(x) dx   \nonumber \\
&&\geq    \int_0^1    \frac{x^{-1}}{ W + c_1 } 2\pi \lambda e^{-\lambda \pi x^2} dx +  \int_1^{\infty}    \frac{x^{-\alpha}}{ W + c_1 } f_{\|\zeta\|}(x) dx  \nonumber \\
&&= \infty,  \nonumber
\end{eqnarray}
\fi
where $f_{\|\zeta\|}(x)  = 2\pi \lambda x e^{-\lambda \pi x^2}$ is the contact distance distribution for the PPP, $(a)$ follows from Jensen's inequality, and $(b)$ follows from Lemma \ref{lemmaIRDirect}. 

So, we only consider the non-singular path loss model.  We have $\mathbb{E}(\textrm{SINR}) = \mathbb{E}(h) \mathbb{E}\big(\frac{\ell(\zeta)}{ W + I(\Phi_o^{\zeta})}\big) \leq \frac{\mathbb{E}(h) }{W}\mathbb{E}(\ell(\zeta)) {<} \infty$.    
Given the ADG $\hat{G}$ of $\Phi$, we have a simple approximation for $M_{\Phi}$: 
\begin{equation}
M_{\Phi} \approx \hat{G} M_{\rm PPP}.
\label{MeanSINReq}
\end{equation}
Therefore, the ADG can also be interpreted as the approximate gain in the mean SINR.

\section{Simulations}
In this section, we present simulation results on a $100 \times 100$ square, where we consider the non-singular path loss model and fix the path loss exponent to $\alpha = 4$ and the intensity of the point processes to $\lambda = 0.1$.  For the MCP, we let $\lambdap = 0.01$, $\bar{c} = 10$ and $\rc = 5$; for the MHP, we let $\lambdab = 0.263$ and $\rh = 1.7$. We present our results in two subsections corresponding to the SINR distribution  and the applications of the ADG. 

\subsection{SINR Distribution}
\subsubsection{Nakagami-$m$ Fading}
In this part, we present simulation results of the outage probability for the PPP, the MCP, and the MHP under Nakagami-$m$ fading. 

\solution
\begin{figure}[hbtp]
\centering
  \includegraphics[width=0.8\figwidth]{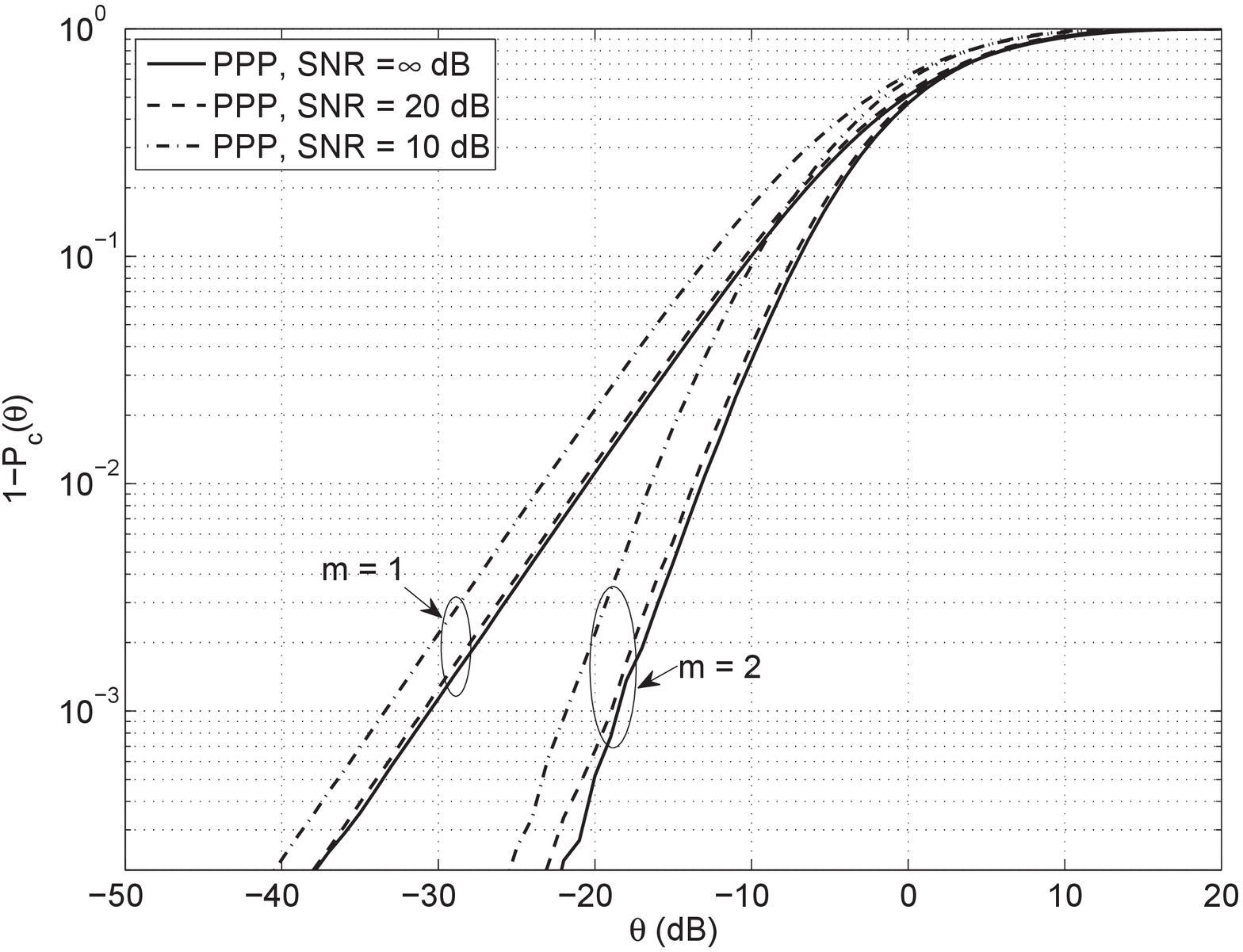} \\ 
  \caption{Nakagami-$m$ fading: the  outage probability $1-\Pc(\theta)$ vs. $\theta$ for the PPP  when $m  \in \{1,2\}$ under different SNR settings.    }
  \label{f1}
\end{figure}
\else
\begin{figure}[hbtp]
\centering
  \includegraphics[width=\figwidth]{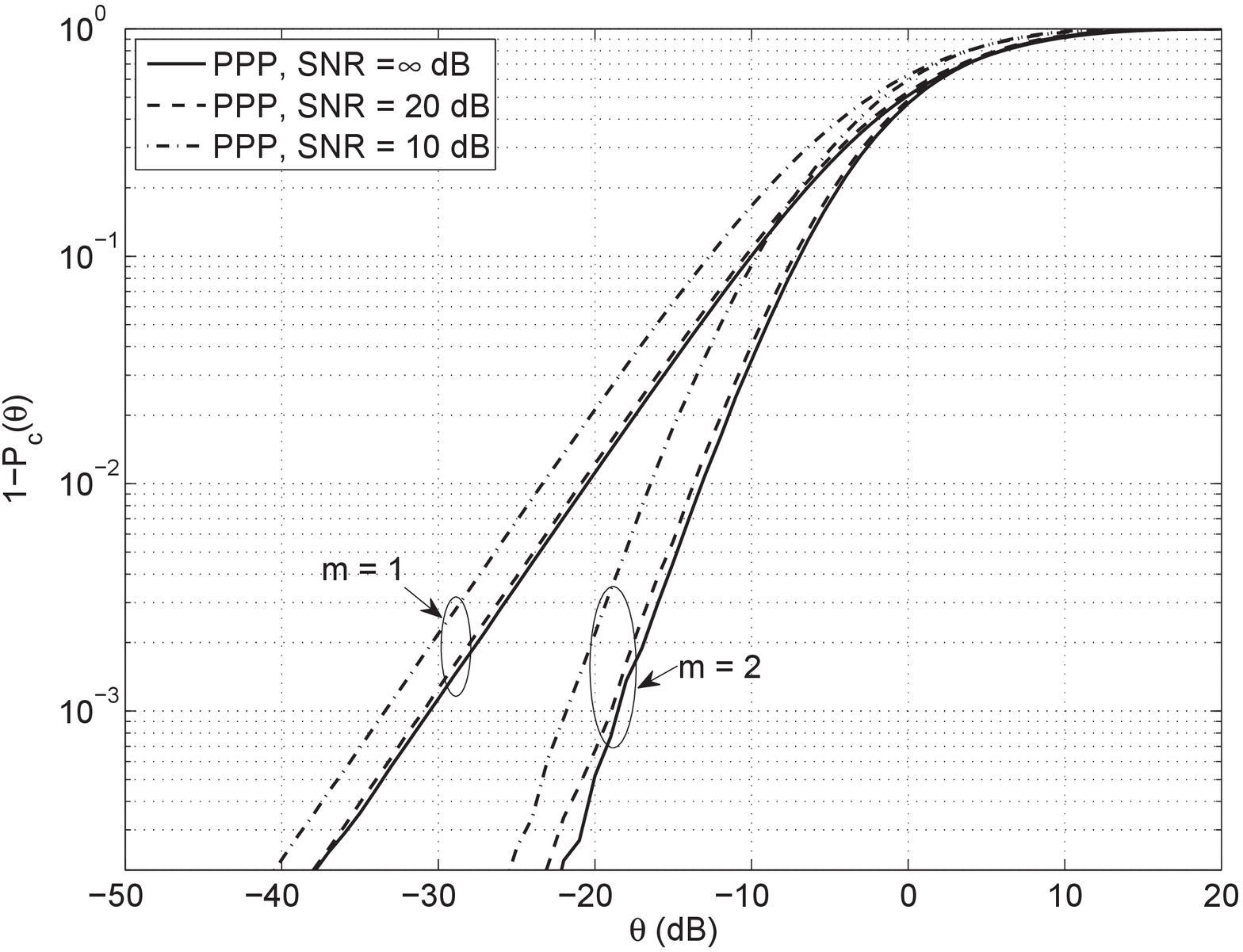} \\ 
  \caption{Nakagami-$m$ fading: the  outage probability $1-\Pc(\theta)$ vs. $\theta$ for the PPP  when $m  \in \{1,2\}$ under different SNR settings.    }
  \label{f1}
\end{figure}
\fi

Fig. \ref{f1} shows the outage curves $1-\Pc(\theta)$ of the PPP for $m \in \{1,2\}$  and different mean SNR values. 
Note that the SNR value here is  
 $1/(2W)$.
As $\theta$ approaches $0$, the slopes of the curves for $m=1$ are all $10$ dB/decade, and the slopes  for $m=2$ are all $20$ dB/decade, in agreement with Corollary \ref{specialcase}. We also observe that
there is only a rather  small gap between the cases of $\textrm{SNR} = 20$ dB and $\textrm{SNR} = \infty$, thus the thermal noise does not significantly  affect the asymptotic performance of the success probability. We will neglect noise in the rest of this section.

\solution
\begin{figure*}
\hspace{-0.5cm}
  \begin{minipage}[t]{0.02\linewidth}~~~
  \end{minipage}
  \begin{minipage}[t]{0.45\linewidth}
    \centering
\includegraphics[width=0.8\figwidth]{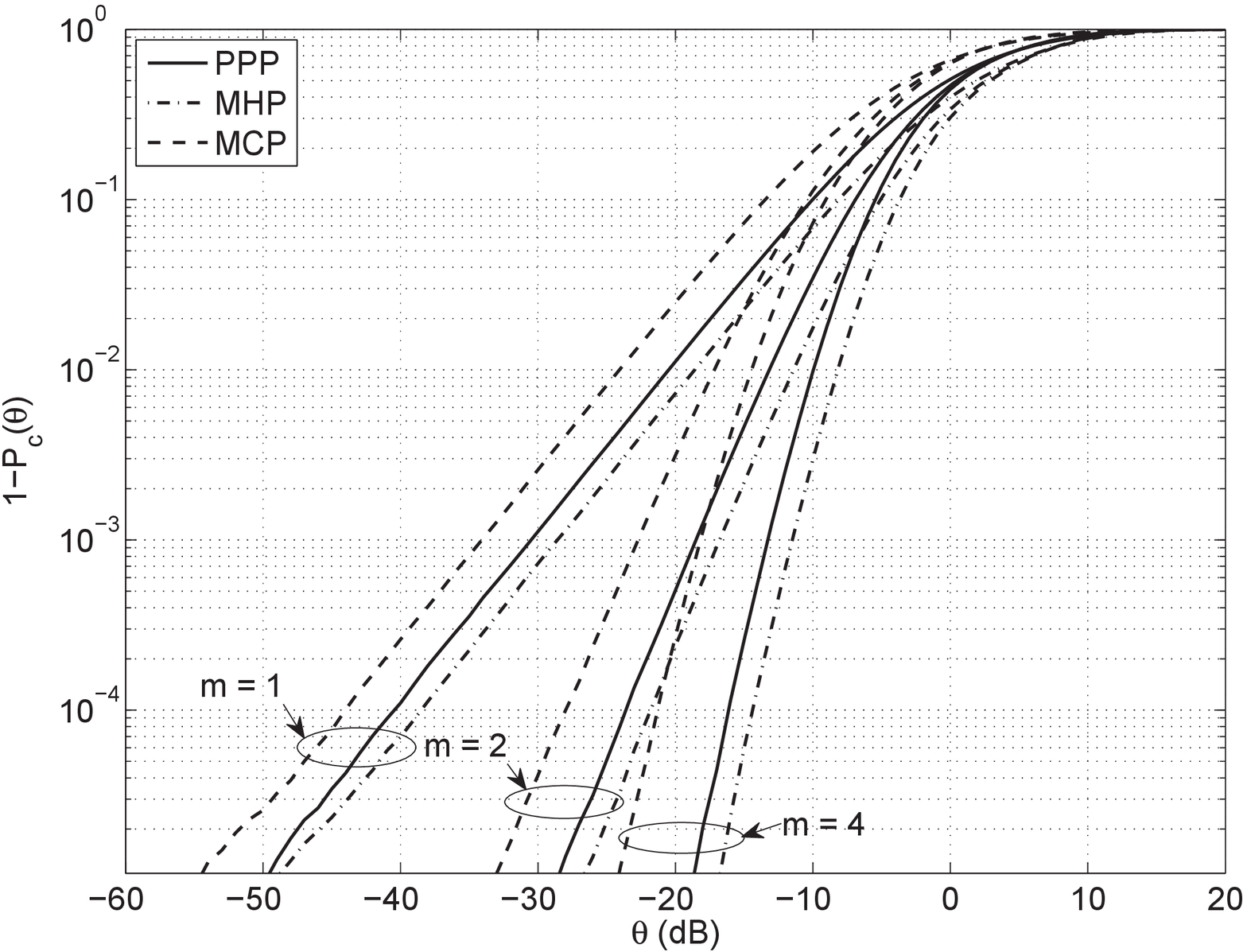}\\
  \caption{Nakagami-$m$ fading:  the  outage probability  $1-\Pc(\theta)$ vs. $\theta$ for the PPP, the MCP and the MHP when $m \in \{1,2,4\}$ (no noise).   }
  \label{f2}
    \end{minipage}
  \begin{minipage}[t]{0.06\linewidth}~
  \end{minipage}
  \begin{minipage}[t]{0.45\linewidth}
    \centering
  \includegraphics[width=0.8\figwidth]{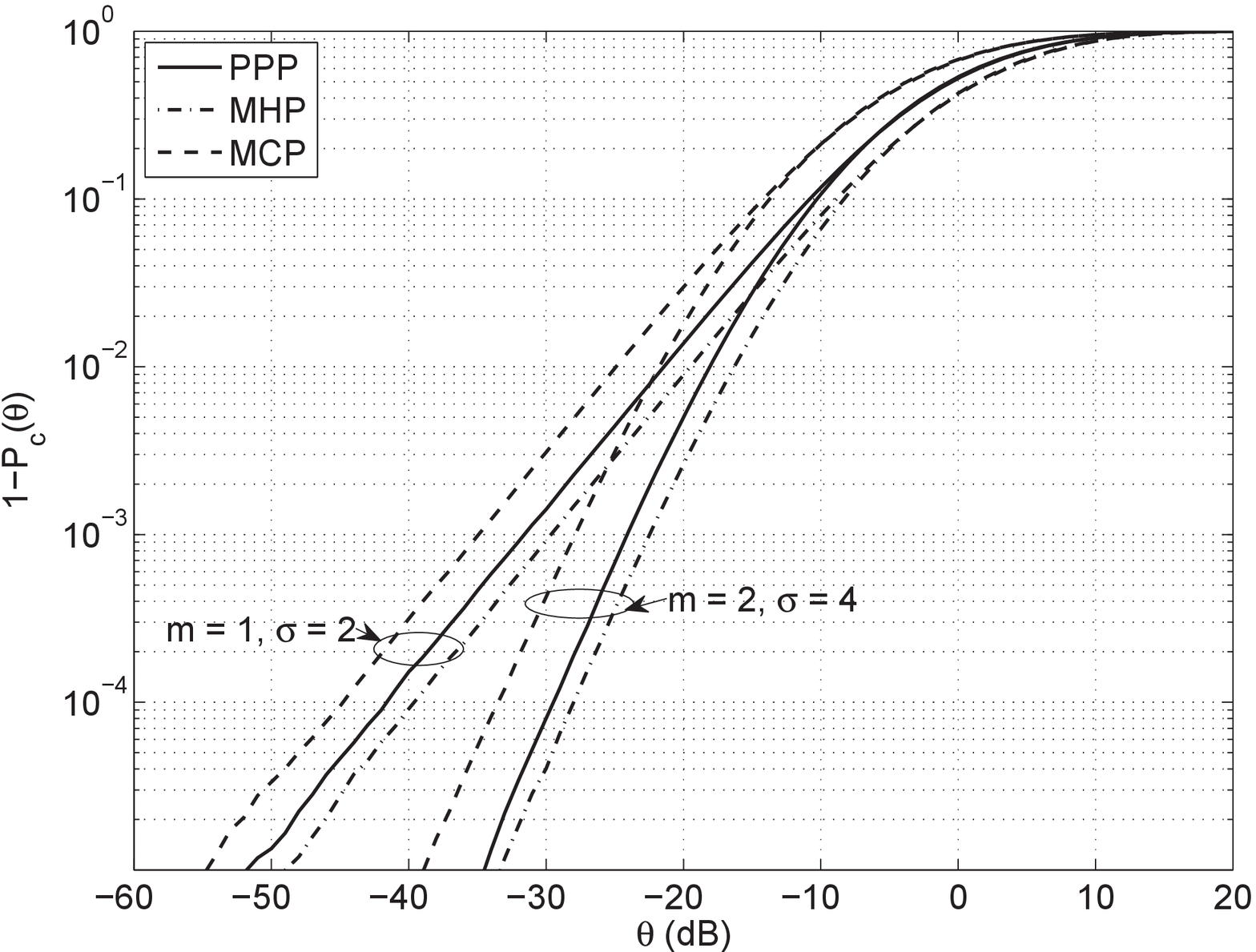}\\  
  \caption{Compound fading: the  outage probability  $1-\Pc(\theta)$ vs. $\theta$ for the PPP, the MCP and the MHP when $m=1, \sigma = 2$ and $m = 2, \sigma = 4$ (no noise, $\alpha = 4$).
  }\label{compound2}
  \end{minipage}%
\end{figure*}
\else
\begin{figure}[hbtp]
\centering
\includegraphics[width=\figwidth]{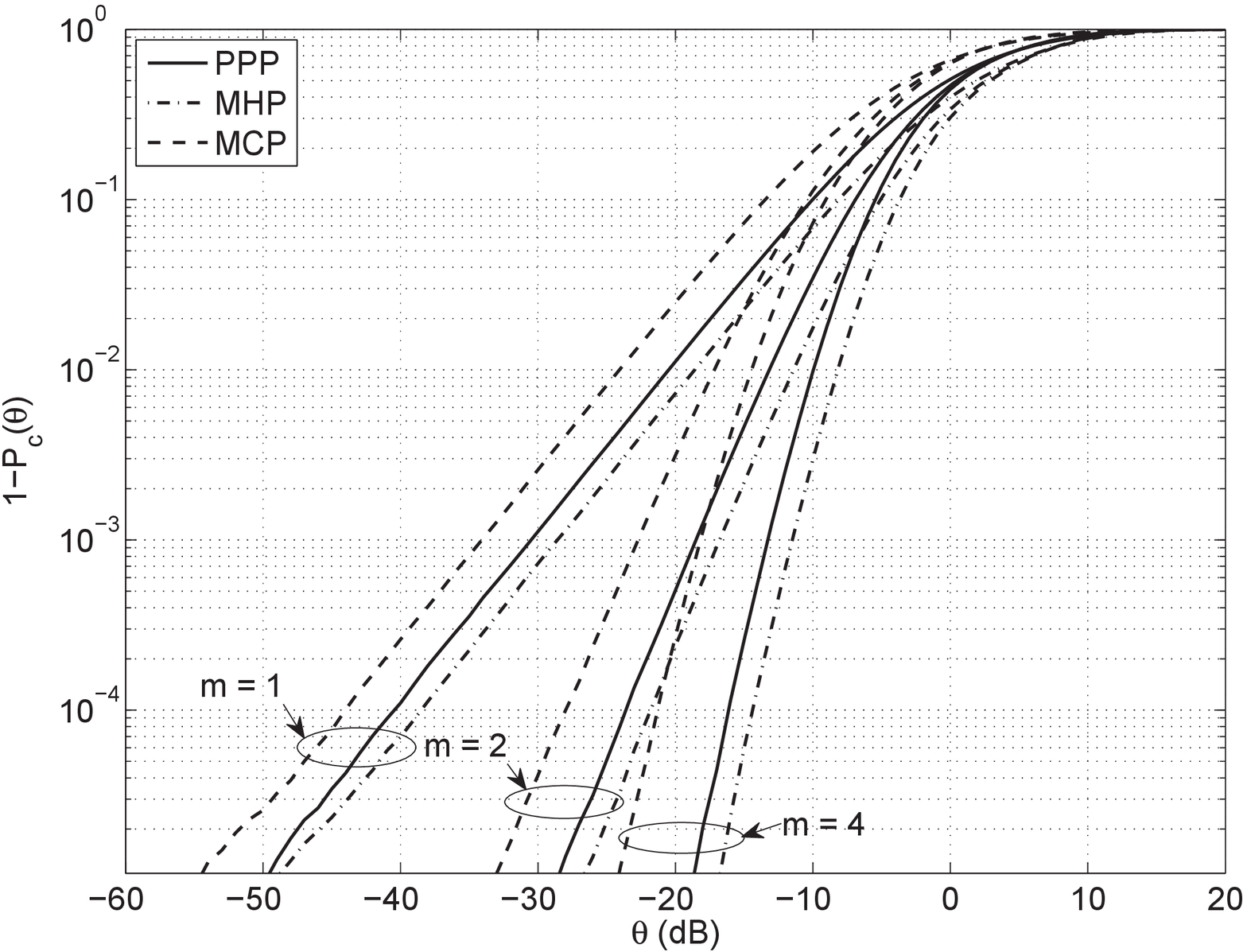}\\
  \caption{Nakagami-$m$ fading:  the  outage probability  $1-\Pc(\theta)$ vs. $\theta$ for the PPP, the MCP and the MHP when $m \in \{1,2,4\}$ (no noise).   }
  \label{f2}
\end{figure}

\begin{figure}[hbtp]
\centering
\includegraphics[width=\figwidth]{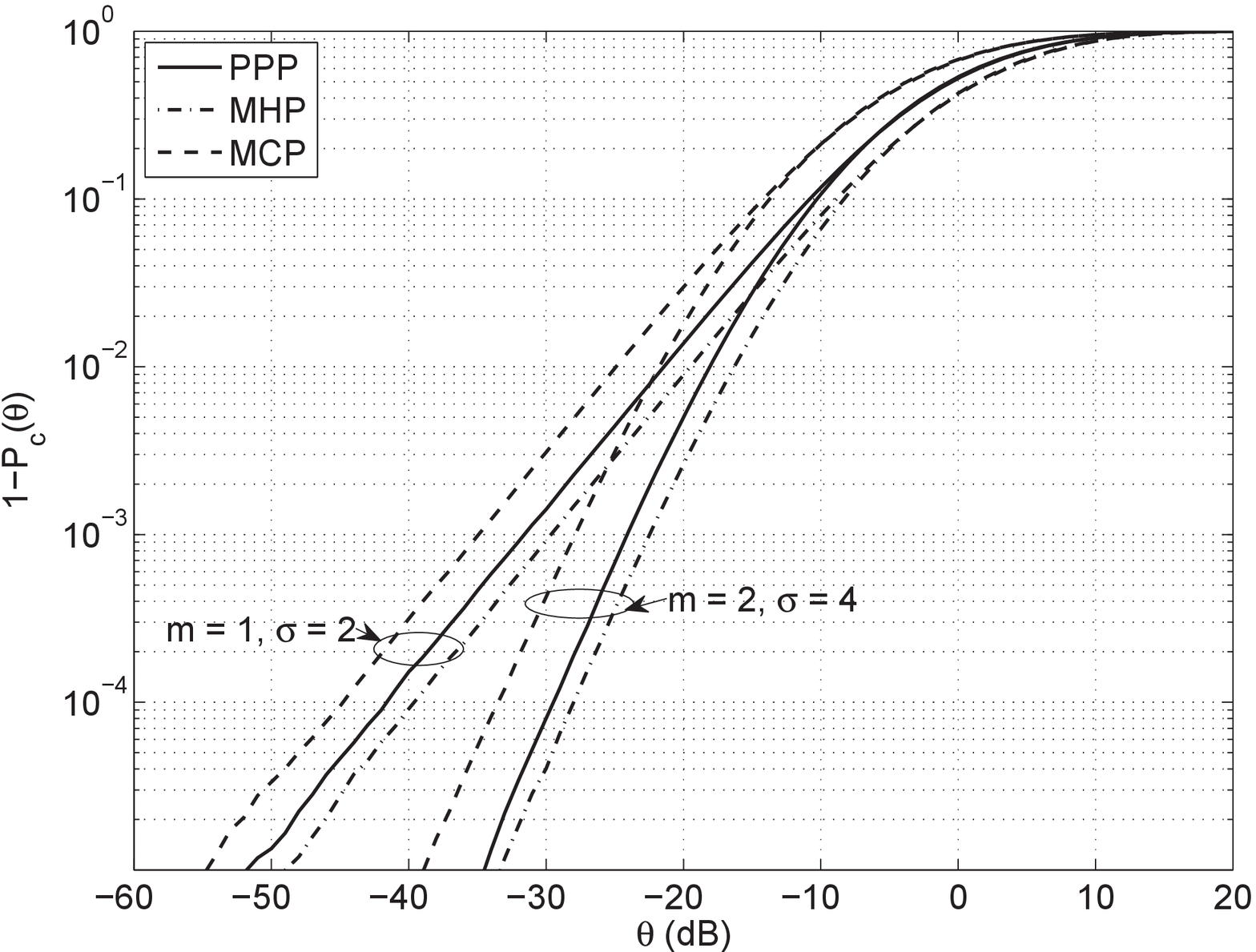}\\  
  \caption{Compound fading: the  outage probability  $1-\Pc(\theta)$ vs. $\theta$ for the PPP, the MCP and the MHP when $m=1, \sigma = 2$ and $m = 2, \sigma = 4$ (no noise, $\alpha = 4$).
  }\label{compound2}
\end{figure}
\fi

In Fig. \ref{f2}, we find that for the same point process, a different $m$ implies a different asymptotic slope. In fact, the slope is $10m$ dB/decade, just as Corollary \ref{specialcase}  indicates. For the same $m$, different point processes have the same asymptotic slope, 
thus in the high-reliability regime, the success  probability of a non-Poisson process can be obtained accurately  simply by shifting the success probability curve of the PPP with the same intensity by the ADG. 
 Besides, we observe that for any $m$, the success probability of the MHP is the largest of  the three processes, followed by the PPP and then the MCP. 
 Intuitively, it is because the MHP is more regular than the PPP and the MCP is more clustered than the PPP. 
In addition, since the value of $\kappa$ for the MCP and the MHP can be approximated through the simulation,  by Corollary \ref{Corol2}, we can approximate their ADGs. Denote by $\hat{G}_{m}^{\rm MCP}$  the ADG for the MCP with respect to $m$, and by $\hat{G}_{m}^{\rm MHP}$ that of the MHP.  
We obtain that for the MCP, $\hat{G}_{1}^{\rm MCP}\approx  0.49$, $\hat{G}_{2}^{\rm MCP}\approx  0.37$ and $\hat{G}_{4}^{\rm MCP}\approx  0.29$; for the MHP, $\hat{G}_{1}^{\rm MHP}\approx  1.58$, $\hat{G}_{2}^{\rm MHP}\approx  1.48$ and $\hat{G}_{4}^{\rm MHP}\approx  1.41$. Note that $\hat{G}_{1}^{\rm MCP}$ is consistent with the approximated value 0.49 obtained from Fig. \ref{f0}.

\subsubsection{Composite Fading}
We consider the combination of Nakagami-$m$ fading and log-normal shadowing in this part.     
In Fig. \ref{compound2}, the outage probabilities for the PPP, the MCP and the MHP are exhibited. The MHP still has the best outage probability, followed by the PPP and the MCP. We also observe that the value of $\sigma$ does not affect the slope of the outage curve as $\theta \to 0$, which is $10m$ dB/decade.
The ADGs of the MCP and MHP can also be estimated: for  $m = 1$ and $\sigma = 2$, $\hat{G}_{1}^{\rm MCP}\approx  0.51$ and $\hat{G}_{1}^{\rm MHP}\approx  1.55$;  for $m = 2$ and $\sigma = 4$, $\hat{G}_{2}^{\rm MCP}\approx  0.40$ and $\hat{G}_{2}^{\rm MHP}\approx  1.37$.

\subsection{Applications of the ADG}
In this subsection, we evaluate the average ergodic rate and the mean SINR  for the PPP, the MCP and the MHP through simulations, and also estimate them using the ADGs. 
The ADG values are approximated by the DG values at $\pt = 1-10^{-4}$ for the three point processes, which are presented in Table \ref{ADGtable}.

\begin{table}[htdp]
\caption{The ADGs for different $\alpha$  (Rayleigh fading, no noise). \mh{Is this in dB?}}
\begin{center}
\begin{tabular}{|c|c|c|c|c|c|} 
 \hline
ADG    &$\alpha = 2.5$& $\alpha = 3.0 $&$\alpha = 3.5$&$\alpha = 4$&$\alpha = 4.5$  \\ \hline
MCP& $0.46$&$0.40$&$0.41$&$0.49$&$0.42$ \\ \hline
MHP&$1.27$&$1.37$&$1.37$&$1.58$&$1.40$ \\ \hline
\end{tabular}
\end{center}
\label{ADGtable}
\end{table}

\subsubsection{Average Ergodic Rate}
In Fig. \ref{AppOfADG_rate}, the  average ergodic rates $\bar{\gamma}$ for the three point processes  as  a function of $\alpha$ are shown as the lines. We also use the simulation results of the PPP and the ADGs in Table \ref{ADGtable} to estimate the average ergodic rates for the MCP and the MHP.  The estimated values are shown as the markers in Fig. \ref{AppOfADG_rate}. From the figure, we observe that  the average ergodic rates estimated using  the ADGs provide fairly good approximations to the empirical values. We also observe that $\bar{\gamma}$ increases as $\alpha$ grows, which is obvious since the interference decays much faster than the desired signal power. 

\subsubsection{Mean SINR}
In Fig. \ref{AppOfADG_meanSINR}, the lines are the mean SINRs for the three point processes as a function of $\alpha$. The markers indicate the mean SINRs for the MCP and the MHP estimated using the simulation results of the PPP and the ADGs. The approximations using the ADGs are acceptable, although not perfect. 
The gaps between the values  estimated using the ADG and the empirical value are  mainly due to the fact that 
the mean is heavily affected by the tail of the CCDF of the SINR, while the ADG approximation is accurate for small and moderate values of $\theta$. 

\solution
\begin{figure*}
\hspace{-0.5cm}
  \begin{minipage}[t]{0.02\linewidth}~~~
  \end{minipage}
  \begin{minipage}[t]{0.45\linewidth}
    \centering
  \includegraphics[width=0.8\figwidth]{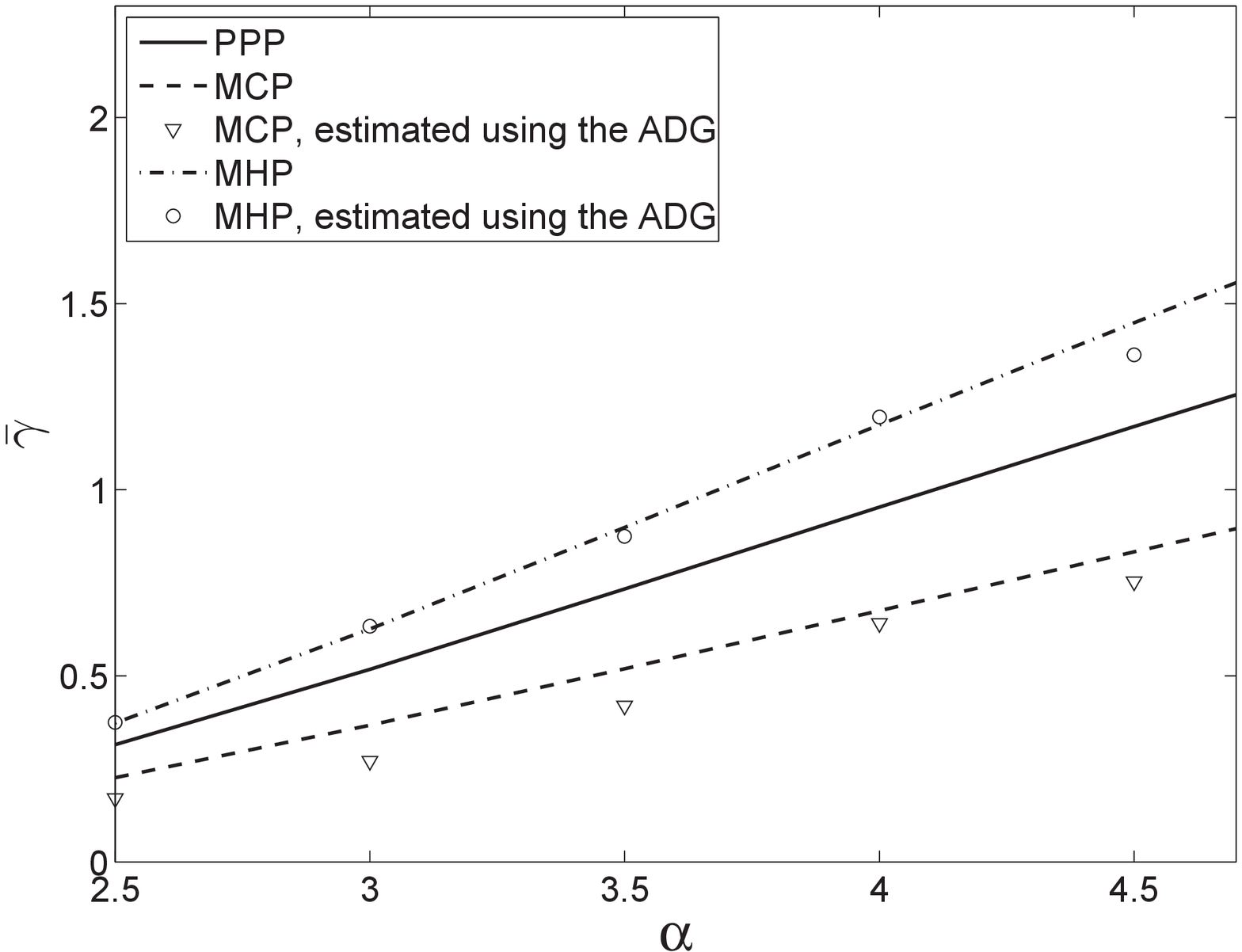}\\
  \caption{The average ergodic rate $\bar{\gamma}$ vs. $\alpha$ for the PPP, the MCP and the MHP.  The lines are the  average ergodic rates obtained directly from simulations, while the markers are the  average ergodic rates estimated using the ADGs. }
  \label{AppOfADG_rate}
  \end{minipage}
  \begin{minipage}[t]{0.06\linewidth}~
  \end{minipage}
  \begin{minipage}[t]{0.45\linewidth}
    \centering
   \includegraphics[width=0.8\figwidth]{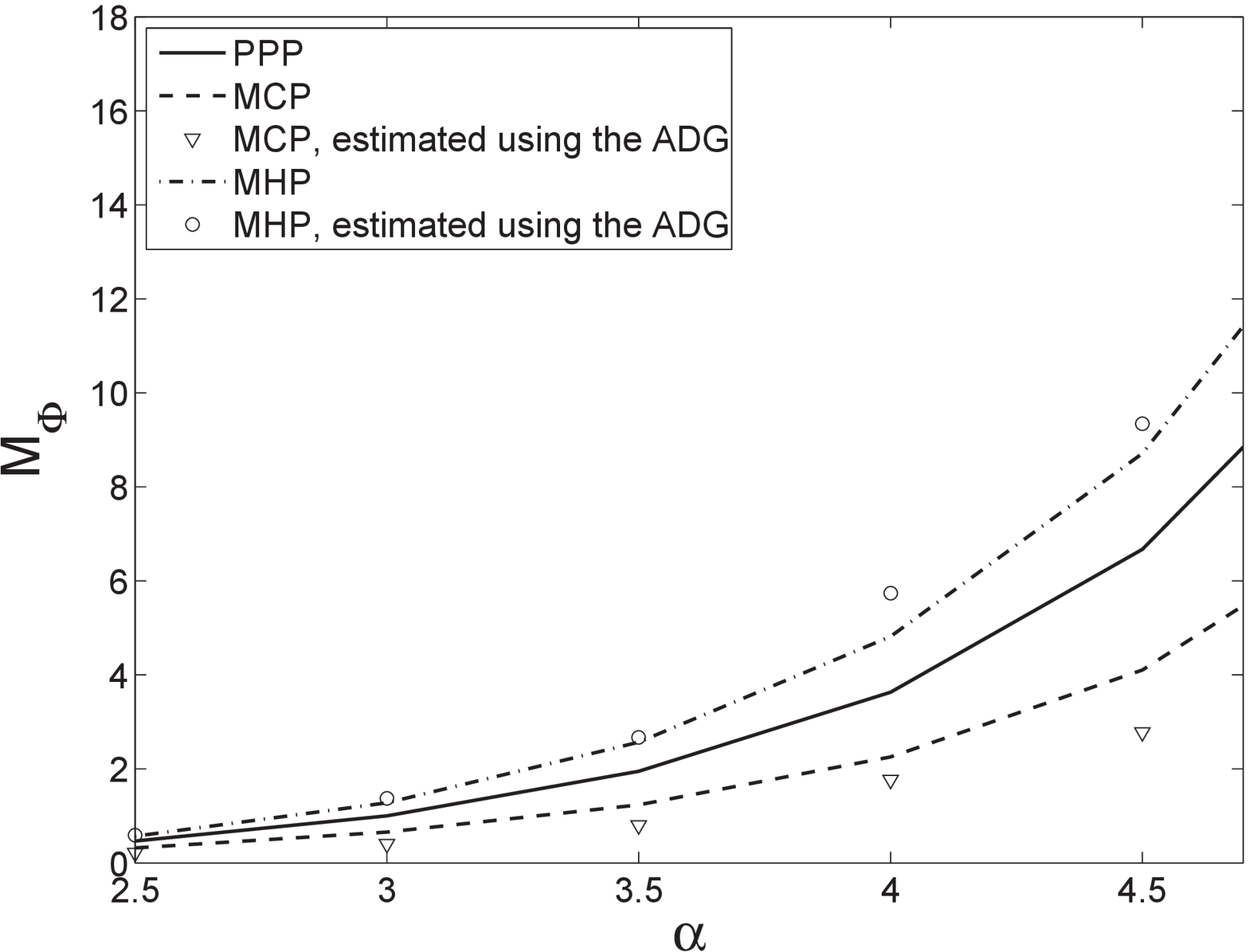}\\
  \caption{The mean SINR $M_{\Phi}$ vs. $\alpha$ for the PPP, the MCP and the MHP.  The lines are the  mean SINRs obtained directly from simulations, while the markers are the  mean SINRs estimated using the ADGs  (i.e., by \eqref{MeanSINReq}). }
  \label{AppOfADG_meanSINR}
  \end{minipage}%
\end{figure*}
\else
\begin{figure}[hbtp]
\centering
  \includegraphics[width=\figwidth]{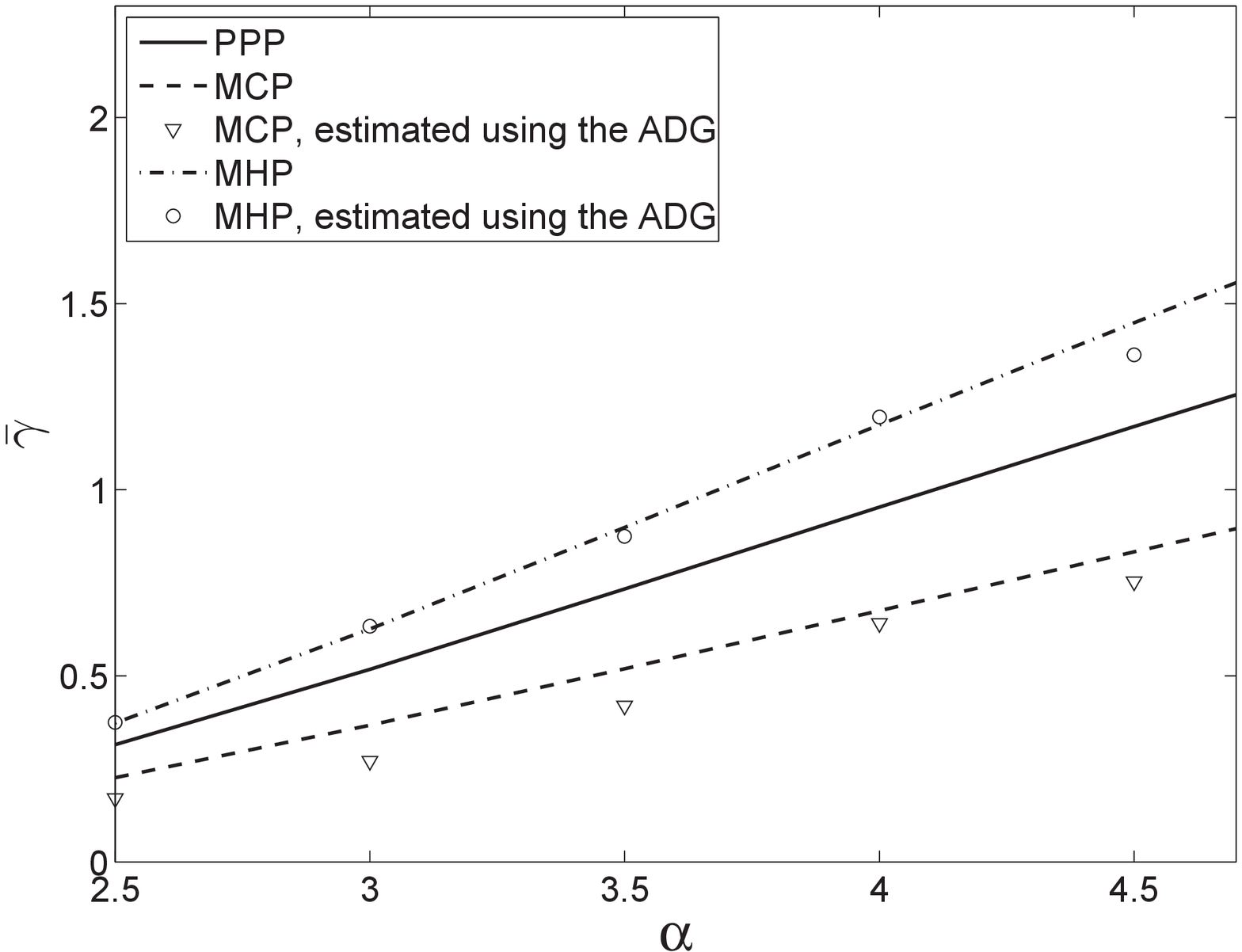}\\
  \caption{The average ergodic rate $\bar{\gamma}$ vs. $\alpha$ for the PPP, the MCP and the MHP.  The lines are the  average ergodic rates obtained directly from simulations, while the markers are the  average ergodic rates estimated using the ADGs. }
  \label{AppOfADG_rate}
\end{figure}

\begin{figure}[hbtp]
\centering
   \includegraphics[width=\figwidth]{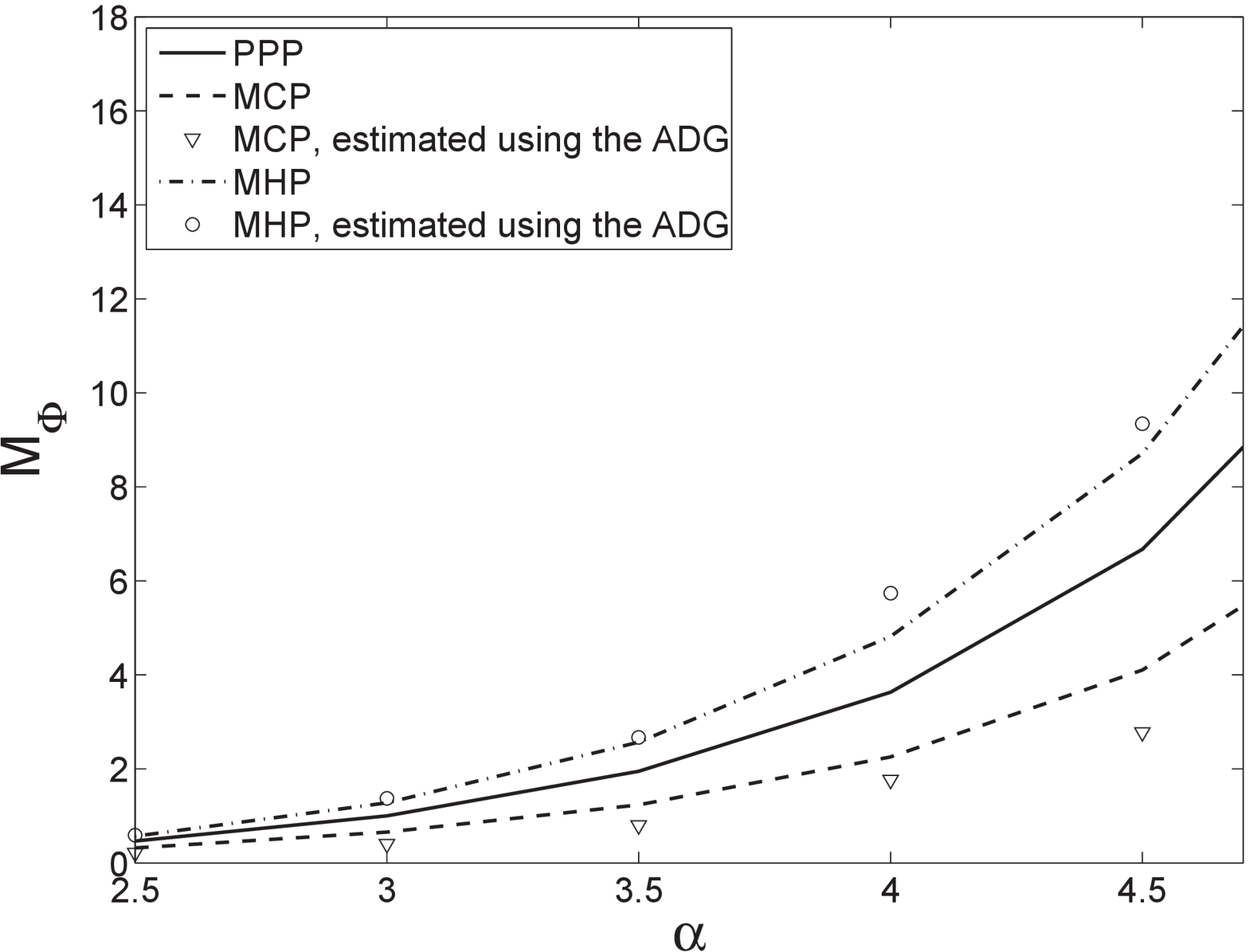}\\
  \caption{The mean SINR $M_{\Phi}$ vs. $\alpha$ for the PPP, the MCP and the MHP.  The lines are the  mean SINRs obtained directly from simulations, while the markers are the  mean SINRs estimated using the ADGs  (i.e., by \eqref{MeanSINReq}). }
  \label{AppOfADG_meanSINR}
\end{figure}
\fi

\section{Conclusions}
In this paper, we examined  the asymptotic properties of the SINR distribution for a variety of motion-invariant  point processes, given some general assumptions on the point process  and general  fading assumptions. 
The assumptions on the point process are satisfied by many commonly used point processes, e.g. the PPP, the MHP and the MCP. Similarly, the fading assumptions  are satisfied by Nakagami-$m$ fading and composite fading. 
We   proved
that   $1-\Pc(\theta) \sim \kappa \theta^m$, as $\theta \to 0$,   
which shows that   the ADG exists. 

Under the same system configurations on the fading and path loss, different point processes with the same intensity have different ADGs.
Thus, the ADG can be used as a simple metric to characterize
the success probability. 
Given the ADG of a point process, we can obtain the precise CCDF of the SINR near 1 by shifting the success probability  curve of the PPP with the same intensity by the ADG (in dB), and numerical studies show that the shifted success probability  curve is highly accurate for all practical success probabilities.

\appendices

\section{Proof of Lemma \ref{lemmaIRDirect}}
\label{sec:appE}
\begin{proof}
We first prove  that
$\forall n \in \mathbb{N}$, there exists a positive $K_0 < \infty$,  s.t. $ \mathbb{E}(\Iy^n) \leq K_0 \mathbb{E}(\IzII^n)$.
Let $\zeta \in \mathbb{R}^2$ and $\|\zeta\| = y$.
According to Def.~\ref{setA}, for $y > y_0$, $\mathbb{P}(\Iy>z) \leq \mathbb{P}(\IzII>z)$,  $\forall z \geq 0$, 
hence  $ \mathbb{E}(\Iy^n) \leq  \mathbb{E}(\IzII^n)$.
For $y \leq y_0$, we have
\begin{align}
  \mathbb{E}(\IzII^n)  &\geq  \mathbb{E}(\IzII^n \mid \Phi^{\zeta}(b(o,y)) = 0) )   \mathbb{P}(\Phi^{\zeta}(b(o,y)) = 0) \nonumber \\
&\stackrel{(a)}{\geq} \mathbb{E}(\Iy^n)\mathbb{P}(\Phi^{\zeta'}(b(o,y_0)) = 0),
\label{fdsas2}
\end{align}
where $\zeta' \in \mathbb{R}^2$, $\|\zeta'\|=y_0$, and $(a)$ holds since $\Phi$ is motion-invariant, $y_0\geq y$ and thus $\mathbb{P}(\Phi^{\zeta'}(b(o,y_0)) = 0) \leq  \mathbb{P}(\Phi^{\zeta}(b(o,y)) = 0)$. The second condition in Def.~\ref{setA} implies that for all $y>0$, $\forall \zeta \in \mathbb{R}^2$ with $\|\zeta\| = y$, $\mathbb{P}(\Phi^{\zeta}(b(o,y))=0 ) \neq 0$.
So, we have $\mathbb{P}(\Phi^{\zeta'}(b(o,y_0)) = 0) \neq 0$,   letting $K_0 = \max\{1, {1}/{\mathbb{P}(\Phi^{\zeta'}(b(o,y_0)) = 0)}\}$, we have
\begin{equation}
 \mathbb{E}(\Iy^n) \leq K_0 \mathbb{E}(\IzII^n). 
\end{equation}

Second, we prove that all moments of $\Iy$ are bounded. 
For $n=1$, by the third  condition in Def.~\ref{setA}, we have
\solution
\begin{align}
\mathbb{E}(\Iy) \leq K_0\mathbb{E}(\IzII) &=  K_0 \mathbb{E}_h\mathbb{E}^{!\zeta}\bigg(\sum_{x\in \Phi \bigcap  B_{\zeta/2}  }h_x \ell(x) \bigg)\nonumber\\
 			 &=K_0\mathbb{E}^{!\zeta}\bigg(\sum_{x\in \Phi \bigcap  B_{\zeta/2}} \mathbb{E}(h_x)\ell(x) \bigg) \nonumber\\
			 &\stackrel{(a)}{=} K_0 \frac{\mathbb{E}(h)}{\lambda}\int_{B_{\zeta/2}} \ell(x) \rho^{(2)}(x-\zeta)dx, 
\label{tmp11}
\end{align}
\else
\begin{eqnarray}
&&\mathbb{E}(\Iy) \leq K_0\mathbb{E}(\IzII)  \nonumber \\
&&=  K_0 \mathbb{E}_h\mathbb{E}^{!\zeta}\bigg(\sum_{x\in \Phi \bigcap  B_{\zeta/2}  }h_x \ell(x) \bigg)\nonumber\\
 &&=K_0\mathbb{E}^{!\zeta}\bigg(\sum_{x\in \Phi \bigcap  B_{\zeta/2}} \mathbb{E}(h_x)\ell(x) \bigg) \nonumber\\
&&\stackrel{(a)}{=} K_0 \frac{\mathbb{E}(h)}{\lambda}\int_{B_{\zeta/2}} \ell(x) \rho^{(2)}(x-\zeta)dx, 
\label{tmp11}
\end{eqnarray}
\fi
where  $\mathbb{E}^{!\zeta}(\cdot)$ is the expectation with respect to the reduced Palm distribution $P^{!\zeta}$, which is the conditional expectation conditioned on $\zeta \in \Phi$ 
but excluding $\zeta$. $(a)$ follows from the  Campbell-Mecke theorem.

For $n \geq 2$, we have
\solution
\begin{align}
\mathbb{E}(\Iy^n)  &\leq  K_0 \mathbb{E}_h\mathbb{E}^{!\zeta}\bigg(\sum_{x\in \Phi \bigcap  B_{\zeta/2} } h_x\ell(x)   \bigg)^n\nonumber\\
			     &\stackrel{(a)}{=} K_0\mathbb{E}_h\mathbb{E}^{!\zeta}\bigg[\sum_{x\in \Phi \bigcap  B_{\zeta/2} }\big(h_x \ell(x)\big)^n \bigg]  + K_0 \sum_{k_1+k_2 = n,k_1\geq k_2 >0} {n \choose k_1,k_2} \nonumber\\
			     &\quad \cdot \mathbb{E}_h\mathbb{E}^{!\zeta}\bigg[\sum^{\neq}_{x_1,x_2\in \Phi \bigcap  B_{\zeta/2} }\big(h_{x_1}\ell(x_1) \big)^{k_1}  \big(h_{x_2}\ell(x_2) \big)^{k_2} \bigg]  + \cdots \nonumber\\
			     &\quad   
			     + K_0\sum_{\sum_{j=1}^{n}k_j = n,k_n\geq \cdots \geq k_1 >0} {n \choose k_1,...,k_n} \mathbb{E}_h					 \mathbb{E}^{!\zeta}\bigg[\sum^{\neq}_{x_1,...,x_n\in \Phi \bigcap  B_{\zeta/2} }\prod_{j=1}^{n}\big(h_{x_j} \ell(x_j) \big)^{k_j} \bigg]  \nonumber\\
			     &\stackrel{(b)}{=}  K_0\frac{\mathbb{E}(h^n)}{\lambda}\int_{B_{\zeta/2} }( \ell(x))^{n} \rho^{(2)}(x-\zeta)dx   + \frac{K_0}{\lambda}\sum_{J = 2}^{n}\sum_{\sum_{j=1}^{J}k_j = n,k_J\geq \cdots \geq k_1 >0} {n \choose k_1,...,k_J}
			       \nonumber\\
			     &\quad \cdot \bigg(  \prod_{j=1}^{J}\mathbb{E}(h^{k_j})\bigg) \int_{B_{\zeta/2}}\cdots\int_{B_{\zeta/2}} \prod_{j=1}^{J}( \ell(x_j))^{k_j} \rho^{(J+1)}(x_1-\zeta,...,x_J-\zeta) dx_1...dx_J, 
\label{tmp12}
\end{align}
\else
\begin{eqnarray}
&&\mathbb{E}(\Iy^n)  \leq  K_0 \mathbb{E}_h\mathbb{E}^{!\zeta}\bigg(\sum_{x\in \Phi \bigcap  B_{\zeta/2} } h_x\ell(x)   \bigg)^n\nonumber\\
			     &&\stackrel{(a)}{=} K_0\mathbb{E}_h\mathbb{E}^{!\zeta}\bigg[\sum_{x\in \Phi \bigcap  B_{\zeta/2} }\big(h_x \ell(x)\big)^n \bigg]  \nonumber \\
			     && \quad + K_0 \sum_{k_1+k_2 = n,k_1\geq k_2 >0} {n \choose k_1,k_2} \nonumber\\
			     &&\quad \cdot \mathbb{E}_h\mathbb{E}^{!\zeta}\bigg[\sum^{\neq}_{x_1,x_2\in \Phi \bigcap  B_{\zeta/2} }\big(h_{x_1}\ell(x_1) \big)^{k_1}  \big(h_{x_2}\ell(x_2) \big)^{k_2} \bigg]   \nonumber\\
			     &&\quad  + \cdots 
			     + K_0\sum_{\sum_{j=1}^{n}k_j = n,k_n\geq \cdots \geq k_1 >0} {n \choose k_1,...,k_n} \nonumber \\
			     && \quad  \cdot \mathbb{E}_h					 \mathbb{E}^{!\zeta}\bigg[ \sum^{\neq}_{x_1,...,x_n\in \Phi \bigcap  B_{\zeta/2} }\prod_{j=1}^{n}\big(h_{x_j} \ell(x_j) \big)^{k_j} \bigg]  \nonumber\\
			     &&\stackrel{(b)}{=}  K_0\frac{\mathbb{E}(h^n)}{\lambda}\int_{B_{\zeta/2} }( \ell(x))^{n} \rho^{(2)}(x-\zeta)dx  \nonumber\\
			     && \quad + \frac{K_0}{\lambda}\sum_{J = 2}^{n}\sum_{\sum_{j=1}^{J}k_j = n,k_J\geq \cdots \geq k_1 >0} {n \choose k_1,...,k_J}
			       \nonumber\\
			     &&\quad \cdot \bigg(  \prod_{j=1}^{J}\mathbb{E}(h^{k_j})\bigg) \int_{B_{\zeta/2}}\cdots\int_{B_{\zeta/2}} \prod_{j=1}^{J}( \ell(x_j))^{k_j} \nonumber \\
			     && \quad \cdot \rho^{(J+1)}(x_1-\zeta,...,x_J-\zeta) dx_1...dx_J, 
\label{tmp12}
\end{eqnarray}
\fi
where $(a)$ follows by the multinomial theorem and $(b)$  follows by the  Campbell-Mecke theorem. 

We discuss the cases of the non-singular and singular path loss models, separately. 
For $\ell(x) = ( 1+\| x\|^{\alpha})^{-1}$, when $n = 1$, since by Def.~\ref{setA}, there exists $q_2<\infty$, such that
$\rho^{(2)}(x)< q_2 \textrm{ for } x \in \mathbb{R}^2$, it yields that $\int_{B_{\zeta/2}} \ell(x) \rho^{(2)}(x)dx \leq \int_{\mathbb{R}^2} \ell(x) \rho^{(2)}(x)dx < \infty$ and thus by \eqref{tmp11}, there exists $c_1 \in \mathbb{R}^+$, such that  $\mathbb{E}(\Iy)<c_1$. Similarly, when $n>1$,  by \eqref{tmp12}, there exists $c_n \in \mathbb{R}^+$, such that $\mathbb{E}(\Iy^n) < c_n$, where $c_n$ does not depend on $\zeta$.

For $\ell(x) = \| x\|^{-\alpha}$, when $n = 1$, we have that $\int_{B_{\zeta/2}} \ell(x) \rho^{(2)}(x)dx \leq q_2 \int_{B_{\zeta/2}} \|x\|^{-\alpha}dx = \frac{2\pi q_2}{(\alpha-2)2^{2-\alpha}}\|\zeta\|^{2-\alpha} \leq \frac{2\pi q_2}{(\alpha-2)2^{2-\alpha}} \max\{1,\|\zeta\|^{2-\alpha}\}$, and hence by \eqref{tmp11}, there exists $c_1 \in \mathbb{R}^+$, such that $\mathbb{E}(\Iy)<c_1 \max\{1,\|\zeta\|^{2-\alpha}\}$.  
When $n > 1$,  for $k_j \in \{1,2,...,n\}$,  $\int_{B_{\zeta/2}} (\ell(x))^{k_j}  dx =  \int_{B_{\zeta/2}} \|x\|^{-\alpha {k_j}}dx = \frac{2\pi}{(\alpha{k_j}-2)2^{2-\alpha{k_j}}}\|\zeta\|^{2-\alpha{k_j}}$,  
and therefore $ \int_{B_{\zeta/2}}\cdots\int_{B_{\zeta/2}} \prod_{j=1}^{J}( \ell(x_j))^{k_j} dx_1...dx_J  = (\prod_{j=1}^{J}(\frac{2\pi}{(\alpha{k_j}-2)2^{2-\alpha{k_j}}} ))\|\zeta\|^{2J-\alpha n} $. Further, we have $ \|\zeta\|^{2J-\alpha n} \leq  \max\{1,\|\zeta\|^{2-\alpha n}\} $. Hence, by \eqref{tmp12}, there
 exists $c_n \in \mathbb{R}^+$, such that $\mathbb{E}(\Iy^n)<c_n \max\{1,\|\zeta\|^{2-\alpha n}\} $.

\end{proof}

\section{Proof of Theorem \ref{thmain}}
\label{proof:thmain}
\begin{proof}
We first consider the case when the noise power $W = 0$.  
Since $\Phi$ is m.i., we can assume $\zeta = (y,0)$.  
Let $\hat{\ell}(x) = 1/\ell(x)$. 
The success probability is
\begin{align}
\Pc(\theta)  & = \mathbb{E}_{\xi}[\mathbb{P}(\textrm{SINR}>\theta \mid {\xi})]   \nonumber \\
& = \int_0^{\infty}\mathbb{P}( h_{\zeta}>\theta\hat{\ell}(\zeta)\Iy  ) f_{\xi}(y) dy \nonumber \\
& = \int_0^{\infty} \mathbb{E}_{\Iy}[F_{h}^{\rm c}(\theta\hat{\ell}(\zeta)\Iy)  ] f_{\xi}(y) dy,
\end{align}
Thus,
 \begin{align}
 \lim_{\theta \to 0} \frac{1-\Pc(\theta)}{\theta^m}
& =  \lim_{\theta \to 0}  \int_0^{\infty}\mathbb{E}_{\Iy}\bigg[ \frac{F_{h}(\theta\hat{\ell}(\zeta)\Iy)}{\theta^m}    \bigg]  f_{\xi}(y) dy.  
\label{VIP1}
\end{align}

Assume $G(t) \triangleq  {F_h(t)}/{t^m}$, for $t>0$, and $G(0)  = \lim_{t \to 0}{F_h(t)}/{t^m} = a $.   
$\forall \epsilon >0$, there exists $\tau > 0$, such that for all $t \in (0,\tau)$, $|G(t)-a|< \epsilon$.
So, $G(t)< a +\epsilon$ for $t \in (0,\tau)$. For $t \geq \tau$, $G(t) =  {F_h(t)}/{t^m}<  \tau^{-m}$. Letting $A = \max \{a+\epsilon,  \tau^{-m} \}$, we have  $G(t)< A$, for all $t\geq0$.

In the following, we discuss the cases of $\ell(x) = ( 1+\| x\|^{\alpha})^{-1}$ and $\ell(x) =  \| x\|^{-\alpha}$, separately. 

For $\ell(x) = ( 1+\| x\|^{\alpha})^{-1}$,
by Lemma \ref{lemmaIRDirect}, we have that    
$\forall n \in \mathbb{N}$, $\exists c_n \in \mathbb{R}^+$, such that $\mathbb{E}(\Iy^n) < c_n$.
It follows that
\solution
 \begin{align}
H(y) &\triangleq \mathbb{E}_{\Iy}\bigg[ \frac{F_{h}(\theta\hat{\ell}(\zeta)\Iy)}{\theta^m}    \bigg]    <  \mathbb{E}_{\Iy} \bigg[ A (\hat{\ell}(\zeta)\Iy)^m    \bigg]    <   A c_m\hat{\ell}(y)^m < +\infty,    
 \label{IRwoutW}
\end{align}
\else
\begin{align}
H(y) &\triangleq \mathbb{E}_{\Iy}\bigg[ \frac{F_{h}(\theta\hat{\ell}(\zeta)\Iy)}{\theta^m}    \bigg]    <  \mathbb{E}_{\Iy} \bigg[ A (\hat{\ell}(\zeta)\Iy)^m    \bigg]    \nonumber \\
&<   A c_m\hat{\ell}(y)^m < +\infty,    
 \label{IRwoutW}
\end{align}
\fi
and thus, by the fourth condition in Def.~\ref{setA}, 
 \begin{align}
\int_0^{\infty} H(y)  f_{\xi}(y) dy  
& <  A c_m \mathbb{E}_{\xi}\big(\hat{\ell}(\xi)^m \big)   < +\infty.  
 \label{IRwoutW2}
\end{align}

For $\ell(x) =  \| x\|^{-\alpha}$, by Lemma \ref{lemmaIRDirect}, we have that $\forall n \in \mathbb{N}$, $\exists d_n \in \mathbb{R}^+$, such that $\mathbb{E}(\Iy^n)<d_n \max\{1,\|\zeta\|^{2-\alpha n}\} $. Therefore, $H(y)<   A y^{\alpha m} d_m \max\{1,y^{2-\alpha m}\}   < +\infty$, and $\int_0^{\infty} H(y)  f_{\xi}(y) dy < A d_m \mathbb{E}_{\xi}\big( {\xi}^{\alpha m} \max\{1,{\xi}^{2-\alpha m}\}  \big)   \leq A d_m \big(\mathbb{E}_{\xi}( {\xi}^{\alpha m}  ) + \mathbb{E}_{\xi}( {\xi}^2 ) \big) < +\infty$.

Assume $\{ \theta_n \}$ is any sequence that converges to 0. Consider  $\ell(x) = ( 1+\| x\|^{\alpha})^{-1}$. Define 
$\tilde{f}(z) \triangleq a (\hat{\ell}(\zeta)z)^m  f_{\Iy}(z)$, and $\tilde{f}_n(z) \triangleq  \frac{F_{h}(\theta_n\hat{\ell}(\zeta)z)}{\theta_n^m} f_{\Iy}(z)$, where $f_{\Iy}(z)$ is the PDF of $\Iy$. $\{\tilde{f}_n\}$ is a sequence of functions and $\tilde{f}_n \to \tilde{f}$, as $n \to \infty$.
Let $g(z) \triangleq A (\hat{\ell}(\zeta)z)^m f_{\Iy}(z) $. We have that $\tilde{f}_n \leq g$, for all $n$, and 
\eqref{IRwoutW} indicates $g(z)$ is integrable.  
By the Dominated Convergence Theorem, we have $\int_0^{\infty} \tilde{f}(z) dz = \lim_{n\to \infty} \int_0^{\infty} \tilde{f}_n(z) dz$.  Similarly, define 
$\hat{f}(y) \triangleq \mathbb{E}_{\Iy}\big[ a (\hat{\ell}(\zeta)\Iy)^m    \big]  f_{\xi}(y)$, $\hat{f}_n(y) \triangleq \mathbb{E}_{\Iy}\big[ \frac{F_{h}(\theta_n\hat{\ell}(\zeta)\Iy)}{\theta_n^m}    \big]  f_{\xi}(y)$ and $\hat{g}(z) \triangleq A c_m\hat{\ell}(y)^m  f_{\xi}(y)$. By the Dominated Convergence Theorem, we have $\int_0^{\infty} \hat{f}(y) dy = \lim_{n\to \infty} \int_0^{\infty} \hat{f}_n(y) dy$. 
By the same reasoning, the Dominated Convergence Theorem can also be applied twice for the case $\ell(x) =  \| x\|^{-\alpha}$. 
Thus, for both cases of $\ell(x)$, we obtain that 
 \begin{align}
 \lim_{\theta \to 0} \frac{1-\Pc(\theta)}{\theta^m}  
& =    \int_0^{\infty}\mathbb{E}_{\Iy}\bigg[\lim_{\theta \to 0}  \frac{F_{h}(\theta\hat{\ell}(\zeta)\Iy)}{\theta^m}     \bigg]   f_{\xi}(y) dy      \nonumber \\
& =    \int_0^{\infty} \mathbb{E}_{\Iy} \bigg[ a  (\hat{\ell}(\zeta)\Iy)^m    \bigg]   f_{\xi}(y) dy. \label{eqfinite}
\end{align}
Note that by \eqref{IRwoutW2}, \eqref{eqfinite} is finite. 

Next,  we consider  the case when $W > 0$.
In \eqref{VIP1}, we only need to replace $\Iy$ with $(\Iy+W)$ in the expectation  $\mathbb{E}_{\Iy}(\cdot)$ 
and the expectation  becomes
\begin{align}
H(y) &= \mathbb{E}_{\Iy}\bigg[ \frac{F_{h}(\theta\hat{\ell}(\zeta)(\Iy+W))}{\theta^m}   \bigg]   \nonumber \\
 &<  \mathbb{E}_{\Iy} \bigg[ A\hat{\ell}(\zeta)^m (\Iy+W)^m    \bigg].
\label{IRwtW}
\end{align}
By expanding $(\Iy+W)^m$, we observe that the right-hand side of (\ref{IRwtW}) is finite.
Analogous to the case when $W = 0$,  we can prove that Theorem \ref{thmain} also holds for $W>0$.
\end{proof}

\section{Proof of Corollary \ref{lemmaIR}}
\label{proof:distrOfI}

\begin{proof}
Consider the worst case,  $F_{h}^{\rm c}(x) \sim \exp(-ax)$, $x \to \infty$. First, we will show  that the Laplace transform of $\Iy$, denoted by $\mathcal{L}_{\Iy}(s)$, converges for $s>\tau_0$, where $\tau_0 < 0$. 
Since $\mathcal{L}_{\Iy}(s)$  always converges for $s \geq 0$, we only consider the case $s < 0$.
To prove the property, we need to derive  an upper bound of  $\mathcal{L}_{\Iy}(s)$ that only depends on the $\Phi^{\zeta}$. 
Similar to the proof of Lemma \ref{lemmaIRDirect}, we can prove the proposition that $\forall s < 0$, there exists a positive $\Ks < \infty$,  s.t. $ \mathbb{E}_{\Iy}(\exp(-s\Iy)) \leq \Ks \mathbb{E}_{\IzII}(\exp(-s\IzII))$.
Thus, we have
\begin{align}
\mathcal{L}_{\Iy}(s)  &= \mathbb{E}_{\Iy}(\exp(-s\Iy))  \nonumber\\
&\leq  \Ks \mathbb{E}_{\Phi^{\zeta}, \{h_x\}} \bigg( \prod_{x\in \Phi^{\zeta} \bigcap B_{\zeta/2} \setminus \{\zeta\}} \exp(-sh_x \ell(x))\bigg) \nonumber\\
&=  \Ks \mathbb{E}^{!\zeta} \bigg( \prod_{x\in \Phi\bigcap B_{\zeta/2}}  \mathbb{E}_h\big(\exp(-sh \ell(x))\big)\bigg) \nonumber\\
& =  \Ks \mathbb{E}^{!\zeta} \bigg( \prod_{x\in \Phi\bigcap B_{\zeta/2}}  \mathcal{L}_h(s \ell(x))\bigg),
\label{PGFL-1}
\end{align}
where $\mathcal{L}_h(s)$ denotes the Laplace transform of $h$.  

Let $k(s,x)  \triangleq   \mathcal{L}_h(s \ell(x))$. We have that  $\mathcal{L}_{\IzII}(s) = \mathbb{E}^{!\zeta} \big( \prod_{x\in \Phi\bigcap B_{\zeta/2} } k(s,x) \big) $ is finite if and only if
$$\eta(s) = \mathbb{E}^{!\zeta} \bigg( \sum_{x\in \Phi\bigcap B_{\zeta/2} }  |\log k(s,x)|\bigg) < \infty.$$
Now we show that $\tau_0$ is strictly less than $0$.  
We have
\begin{align}
\eta(s) &= \mathbb{E}^{!\zeta} \bigg( \sum_{x\in \Phi\bigcap B_{\zeta/2} }  |\log k(s,x)|\bigg)\nonumber\\
	&\stackrel{(a)}{=} \frac{1}{\lambda}\int_{B_{\zeta/2}} \big|\log k(s,x) \big| \rho^{(2)}(x-\zeta)dx,
\end{align}
where $(a)$ follows from  the Campbell-Mecke theorem.  

Since $F_{h}^{\rm c}(x) \sim \exp(-ax)$ for large $x$, without  loss of generality, we assume for some large $H_0$, the PDF of $h$ is $f_{\xi}(x) = a\exp(-ax)$ ($x>H_0$).
So,
\solution
\begin{align}
k(s,x) &= \int_0^{\infty} \exp(-sy \ell(x)) dF_h(y)     \nonumber\\
	&= \int_0^{H_0} \exp(-sy \ell(x)) dF_h(y)    + \int_{H_0}^{\infty} a \exp\big(-y(a+s\ell(x))\big) dy.
\label{notuse1}
\end{align}
\else
\begin{align}
k(s,x) &= \int_0^{\infty} \exp(-sy \ell(x)) dF_h(y)     \nonumber\\
	&= \int_0^{H_0} \exp(-sy \ell(x)) dF_h(y)    \nonumber\\
	& \quad + \int_{H_0}^{\infty} a \exp\big(-y(a+s\ell(x))\big) dy.
\label{notuse1}
\end{align}
\fi

Since $x\in  \Phi\bigcap B_{\zeta/2} $, by the Dominated Convergence Theorem, $k(s,x)$ is bounded for all $x$ and $s>-a \ell(\|\zeta\|/2)^{-1}$.  Also, for $s\in(-a\ell(\|\zeta\|/2)^{-1},0)$, we have $k(s,x)>1$ and   $\log(k(s,x))\leq k(s,x)-1$. To show $\eta(s)<\infty$  for $s\in(-a\ell(\|\zeta\|/2)^{-1},0)$, we need to prove   $\int_{B(o,\omega)^{\rm c}} ( k(s,x)-1) \rho^{(2)}(x)dx<\infty$, for large $\omega$. Since for large $\|x\|$, we have $\rho^{(2)}(x-\zeta)\to \lambda^2$, where $\lambda$ is the intensity of $\Phi$, we choose   $\omega$ large enough such that $\rho^{(2)}(x)$ is approximately $\lambda^2$ for all $\|x\|>\omega$.
So we only need to show that $\int_{B(o,\omega)^{\rm c}} ( k(s,x)-1)dx<\infty$. We have
\solution
\begin{align}
&\quad \int_{B(o,\omega)^{\rm c}} ( k(s,x)-1)dx \nonumber\\
&=   \int_{B(o,\omega)^{\rm c}} \int_0^{H_0} (\exp(-s y\ell(x))-1) dF_h(y)  dx  +  \int_{B(o,\omega)^{\rm c}} \int_{H_0}^{\infty} (\exp(-s y \ell(x))-1) dF_h(y)  dx.  \nonumber
\end{align}
\else
\begin{eqnarray}
&& \int_{B(o,\omega)^{\rm c}} ( k(s,x)-1)dx \nonumber\\
&&=   \int_{B(o,\omega)^{\rm c}} \int_0^{H_0} (\exp(-s y\ell(x))-1) dF_h(y)  dx  \nonumber\\
&&\quad +  \int_{B(o,\omega)^{\rm c}} \int_{H_0}^{\infty} (\exp(-s y \ell(x))-1) dF_h(y)  dx.  \nonumber
\end{eqnarray}
\fi

For large $\omega$,
\solution
\begin{align}
\quad  \int_{B(o,\omega)^{\rm c}} \int_0^{H_0} (\exp(-s y \ell(x))-1) dF_h(y)  dx   = \int_{B(o,\omega)^{\rm c}} \int_0^{H_0} (-s y \ell(x)) dF_h(y)  dx < \infty,   \nonumber
\end{align}
\else
\begin{eqnarray}
&&\int_{B(o,\omega)^{\rm c}} \int_0^{H_0} (\exp(-s y \ell(x))-1) dF_h(y)  dx   \nonumber\\
&&= \int_{B(o,\omega)^{\rm c}} \int_0^{H_0} (-s y \ell(x)) dF_h(y)  dx < \infty,   \nonumber
\end{eqnarray}
\fi
and
\solution
\begin{align}
&\quad  \int_{B(o,\omega)^{\rm c}} \int_{H_0}^{\infty} (\exp(-s y \ell(x))-1) dF_h(y)  dx   \nonumber\\
&= \exp(-aH_0) \int_{B(o,\omega)^{\rm c}} \bigg( \frac{-s}{a\ell(x)+s} + \frac{a\ell(x) (\exp(-s H_0 \ell(x))-1)}{a\ell(x)+s} \bigg) dx  < \infty,   \nonumber
\end{align}
\else
\begin{eqnarray}
&&  \int_{B(o,\omega)^{\rm c}} \int_{H_0}^{\infty} (\exp(-s y \ell(x))-1) dF_h(y)  dx   \nonumber\\
&&= \exp(-aH_0) \int_{B(o,\omega)^{\rm c}} \bigg( \frac{-s}{a\ell(x)+s} \nonumber\\
&& \quad + \frac{a\ell(x) (\exp(-s H_0 \ell(x))-1)}{a\ell(x)+s} \bigg) dx  < \infty,   \nonumber
\end{eqnarray}
\fi

Thus, $\eta(s)<\infty$ and $\mathcal{L}_{\Iy}(s) < \infty$.  Since $\Iy$ is nonnegative, according the region of convergence (ROC) for Laplace transforms, there exists $\tau < -a\ell(\|\zeta\|/2)^{-1}$, such that $\mathcal{L}_{\Iy}(s)$ converges for $s<\tau$ and diverges for $s>\tau$. $\tau$ is called the abscissa of convergence.
By Theorem 3 in \cite{Tauberian}, it follows that the interference has an exponential tail.
Therefore, if the fading has at most an exponential tail,  the interference tail is bounded by an exponential.
\end{proof}

\section{Proof of Lemma \ref{LemmaCompound}}
\label{sec:appF}
\begin{proof}
Since  $\tilde{h}$ and $\hat{h}$ are independent, we have
\solution
\begin{align}
F_h(t) &= \mathbb{P}(\tilde{h} \hat{h} \leq t) = \int_0^{\infty}  \mathbb{P}(\tilde{h} \leq \frac{t}{u} \mid \hat{h} = u) f_{\hat{h}}(u)du  = \int_0^{\infty} F_{\tilde{h}}\Big(\frac{t}{u}\Big) f_{\hat{h}}(u)du   \nonumber\\
&=  \int_0^{\infty}  \frac{1}{\sqrt{\pi} \Gamma(m)} \bigg(   \int_0^{\frac{mt}{u}} w^{m-1} \exp(-w) dw \bigg) \frac{V_{\sigma}}{u} \exp(-V_{\sigma}^2(\log u)^2) du,
\end{align}
\else
\begin{align}
F_h(t) &= \mathbb{P}(\tilde{h} \hat{h} \leq t) = \int_0^{\infty}  \mathbb{P}(\tilde{h} \leq \frac{t}{u} \mid \hat{h} = u) f_{\hat{h}}(u)du   \nonumber\\
&= \int_0^{\infty} F_{\tilde{h}}\Big(\frac{t}{u}\Big) f_{\hat{h}}(u)du   \nonumber\\
&=  \int_0^{\infty}  \frac{1}{\sqrt{\pi} \Gamma(m)} \bigg(   \int_0^{\frac{mt}{u}} w^{m-1} \exp(-w) dw \bigg)   \nonumber\\
&\quad \cdot \frac{V_{\sigma}}{u} \exp(-V_{\sigma}^2(\log u)^2) du,
\end{align}
\fi
where  $V_{\sigma} \triangleq \frac{10}{\sigma \sqrt{2} \log10}$.

To prove the first property, we have
\solution
\begin{align}
\lim_{t \to 0} \frac{F_h(t)}{t^m} &= \lim_{t \to 0} \frac{F_h'(t)}{m t^{m-1}}  
= \lim_{t \to 0}  \int_0^{\infty}  \frac{V_{\sigma} m^{m-1}}{\sqrt{\pi} \Gamma(m) u^{m+1}} \exp\bigg(-\frac{mt}{u}\bigg)  \exp(-V_{\sigma}^2(\log u)^2) du  \nonumber\\
&\leq  \int_0^{\infty}  \frac{V_{\sigma} m^{m-1}}{\sqrt{\pi} \Gamma(m) u^{m+1}}  \exp(-V_{\sigma}^2(\log u)^2) du.
\label{vbnm1}
\end{align}
\else
\begin{eqnarray}
&&\lim_{t \to 0} \frac{F_h(t)}{t^m} = \lim_{t \to 0} \frac{F_h'(t)}{m t^{m-1}}  \nonumber\\ 
&&= \lim_{t \to 0}  \int_0^{\infty}  \frac{V_{\sigma} m^{m-1}}{\sqrt{\pi} \Gamma(m) u^{m+1}} \nonumber\\
&&\quad \cdot \exp\bigg(-\frac{mt}{u}\bigg)  \exp(-V_{\sigma}^2(\log u)^2) du  \nonumber\\
&&\leq  \int_0^{\infty}  \frac{V_{\sigma} m^{m-1}}{\sqrt{\pi} \Gamma(m) u^{m+1}}  \exp(-V_{\sigma}^2(\log u)^2) du.
\label{vbnm1}
\end{eqnarray}
\fi

Since as $u\to 0$, $\exp(-V_{\sigma}^2(\log u)^2) = o(u^n)$ for any $n \in \mathbb{N}$, \eqref{vbnm1} is bounded. Thus we can apply the Dominated Convergence Theorem and obtain the first property. 

For the second property, on the one hand, for any $n \in \mathbb{N}$,
\solution
\begin{align}
\lim_{t \to \infty}\frac{1-F_h(t)}{t^{-n}} &=  \lim_{t \to \infty}\frac{F_h'(t)}{nt^{-n-1}} =  \lim_{t \to \infty}  \int_0^{\infty}  \frac{V_{\sigma} m^{m} t^{n+m}}{\sqrt{\pi} \Gamma(m) u^{m+1}n}  \exp\Big(-\frac{mt}{u}\Big) \exp(-V_{\sigma}^2(\log u)^2) du.   \nonumber
\end{align}
\else
\begin{eqnarray}
&&\lim_{t \to \infty}\frac{1-F_h(t)}{t^{-n}} =  \lim_{t \to \infty}\frac{F_h'(t)}{nt^{-n-1}} \nonumber\\
&& =  \lim_{t \to \infty}  \int_0^{\infty}  \frac{V_{\sigma} m^{m} t^{n+m}}{\sqrt{\pi} \Gamma(m) u^{m+1}n}  \nonumber\\
&& \quad \cdot \exp\Big(-\frac{mt}{u}\Big) \exp(-V_{\sigma}^2(\log u)^2) du.   \nonumber
\end{eqnarray}
\fi

Assume $H(t) = t^{n+m} \exp\big(-\frac{mt}{u}\big)$. Since $H'(t) = t^{n+m-1} \big(n+m - \frac{mt}{u} \big) \exp\big(-\frac{mt}{u}\big)$, when $t = \frac{u(n+m)}{m}$, $H(t)$ achieves its maximum value and $\max_{t>0} H(t) = \big(\frac{u(n+m)}{m} \big)^{n+m} \exp(-(n+m))$. Thus,
\solution
\begin{align}
\lim_{t \to \infty}\frac{1-F_h(t)}{t^{-n}} &\leq   \int_0^{\infty}   \frac{V_{\sigma} u^{n-1}}{\sqrt{\pi} \Gamma(m) n} \frac{(n+m)^{n+m}}{m^{n}}  \exp(-(n+m)) \exp(-V_{\sigma}^2(\log u)^2) du  < \infty.   \nonumber
\end{align}
\else
\begin{align}
\lim_{t \to \infty}\frac{1-F_h(t)}{t^{-n}} &\leq   \int_0^{\infty}   \frac{V_{\sigma} u^{n-1}}{\sqrt{\pi} \Gamma(m) n} \frac{(n+m)^{n+m}}{m^{n}}  \nonumber\\
& \quad \cdot \exp(-(n+m)) \exp(-V_{\sigma}^2(\log u)^2) du   \nonumber\\
&< \infty.   \nonumber
\end{align}
\fi

Applying the Dominated Convergence Theorem, we obtain $\lim_{t \to \infty}\frac{1-F_h(t)}{t^{-n}} = 0$ and thus $F_h^{\rm c}(t) = o(t^{-n})$, as $t \to \infty$, for any $n \in \mathbb{N}$.

On the other hand,  for any $a>0$, 
\solution
\begin{align}
\lim_{t \to \infty}\frac{1-F_h(t)}{\exp(-at)} &=  \lim_{t \to \infty}\frac{F_h'(t)}{a\exp(-at)} \nonumber\\
&=  \lim_{t \to \infty}  \int_0^{\infty}   \frac{V_{\sigma} m^{m}}{\sqrt{\pi} \Gamma(m) u^{m+1}a} t^{m-1} \exp\Big(\Big(a-\frac{m}{u}\Big)t\Big) \exp(-V_{\sigma}^2(\log u)^2) du.
\end{align}
\else
\begin{eqnarray} 
&&\lim_{t \to \infty}\frac{1-F_h(t)}{\exp(-at)} =  \lim_{t \to \infty}\frac{F_h'(t)}{a\exp(-at)} \nonumber\\
&&=  \lim_{t \to \infty}  \int_0^{\infty}   \frac{V_{\sigma} m^{m}}{\sqrt{\pi} \Gamma(m) u^{m+1}a} t^{m-1} \exp\Big(\Big(a-\frac{m}{u}\Big)t\Big) \nonumber\\
&& \quad \cdot \exp(-V_{\sigma}^2(\log u)^2) du.
\end{eqnarray}
\fi

For any $a>0$, there exists $\hat{K}>0$, such that for $u > \hat{K}$, $\exp(m t/u) < \exp({at}/{3})$. Hence,  $\lim_{t \to \infty}\frac{1-F_h(t)}{\exp(-at)}  = \infty$, for any $a>0$.
Thus, $-\log F_h^{\rm c}(t)  = o(t)$, $t \to \infty$.
\end{proof}

\section{Proof of Lemma \ref{lemma3PP}}
\label{sec:appC}

\begin{proof}
Conditions  1 and 2 in Def.~\ref{setA} hold for all the three point processes obviously. For Conditions 3 and 4, we treat the three point processes separately.

For the PPP, Condition 3 holds, because the points in $\Phi$ are independent; Condition 4 holds, because $\mathbb{P}(\xi > x) = \mathbb{P}(\Phi(b(o,x)) = 0) = \exp(-\lambda \pi x^2)$.

For the MCP,  we first prove that Condition 3 holds. 
For $y > \rc$, the interference $\Iy$ consists of two parts. One is the interference from the clusters with center points inside the region $B(o,y+\rc)\setminus b(o,y-\rc)$, denoted by $I_1$, and the other part is the interference from the clusters with center points in $B(o,y + \rc)^{\rm c}$, denoted by $I_2$. $I_1$ and $I_2$ are independent.
Similarly, $\IzII$ consists of $\hat{I}_1$ and $\hat{I}_2$, where  $\hat{I}_1$ is from the clusters with center points inside   $B(o,y+\rc)\setminus b(o,y/2)$ and  $\hat{I}_2$ is  from the clusters with center points in $B(o,y + \rc)^{\rm c}$.

Since the parent points are independent, $I_2$ and $\hat{I}_2$ have the same distribution. For $y \gg \rc$, we can easily prove that $\hat{I}_1$ stochastically dominates $I_1$. 
As $\mathbb{P}(\Iy>z)  =  \mathbb{P}(I_1 + I_2 >z)     =  \mathbb{E}_{I_2}[\mathbb{P}(I_1 > z-I_2 \mid I_2)]$, we have  $\mathbb{P}(\Iy>z) \leq \mathbb{P}(\IzII>z)$ for all $z\geq 0$.

Then we prove Condition 4 holds for the MCP. For large $y$, let $\mathcal{S}$ be the set of the parent points that are in $B(o,y-\rc)$, i.e., $\mathcal{S} = \{ x \in \Phip: x \in B(o,y-\rc)  \}$ and $\tilde{\Phi}_x$ be the daughter process for the cluster centered at $x \in \Phip$. We have 
\solution
\begin{align}
\mathbb{P}(\xi > y) &=  \mathbb{P}(\Phi(B(o,y)) = 0)   \stackrel{(a)}{\leq} \mathbb{P}(\tilde{\Phi}_x(B(x,\rc)) = 0, \textrm{ for all } x\in \mathcal{S})   \nonumber\\
&=  \sum_{k = 0}^{\infty} \frac{(\lambdap \pi (y-\rc)^2)^k  \exp(-\lambdap \pi (y-\rc)^2)}{k!} \exp(-\bar{c})^k   \nonumber\\
&= \exp\big(- (1-\exp(-\bar{c}))\lambdap \pi (y-\rc)^2\big),
\end{align}
\else
\begin{eqnarray}
&&\mathbb{P}(\xi > y) =  \mathbb{P}(\Phi(B(o,y)) = 0)  \nonumber\\
&&   \stackrel{(a)}{\leq} \mathbb{P}(\tilde{\Phi}_x(B(x,\rc)) = 0, \textrm{ for all } x\in \mathcal{S})   \nonumber\\
&&=  \sum_{k = 0}^{\infty} \frac{(\lambdap \pi (y-\rc)^2)^k  \exp(-\lambdap \pi (y-\rc)^2)}{k!} \exp(-\bar{c})^k   \nonumber\\
&&= \exp\big(- (1-\exp(-\bar{c}))\lambdap \pi (y-\rc)^2\big),
\end{eqnarray}
\fi
where $(a)$ follows since $\Phi(B(o,y)) = 0$ implies $\Phi(B(x,\rc)) = 0, \textrm{ for all } x\in \mathcal{S}$.
As $\mathbb{E}(\xi^n) = -\int z^n d\mathbb{P}(\xi>z)$, performing integration by parts, it follows that $\mathbb{E}(\xi^n) $ is bounded. 

For the MHP, to prove Condition 3,
we consider $\Phi^{\zeta}_o$ and $\Phi^{\zeta}$ in term of the base PPP $\Phib$.
Conditioned on $\Phib \bigcap (B(o,y+2\rh)\setminus B(o,y+\rh))$, the interference from the region $B(o,y+2\rh)^{\rm c}$ in $\Phi^{\zeta}$ and that in $\Phi_o^{\zeta}$ are i.i.d.. So we only need to consider the region $B(o,y+2\rh)$ for large $y$. 
As $y$ grows, $\mathbb{E}[\Phi^{\zeta}(B(o,y)\setminus B(o,y/2))] = \Theta(y^2)$, and $\mathbb{E}[ \Phi^{\zeta}_o (B(o,y+2\rh) \setminus B(o,y))] = \Theta(y)$.\footnote{$f(x) = \Theta(g(x))$, if both $f(x)/g(x)$ and $g(x)/f(x)$ remain bounded as $x \to \infty$.} 
It can be proved that  the portion of $\IzII$    that comes from the retained points
 in $B(o,y+2\rh) \setminus B(o,y/2)$
stochastically dominates  the portion of $\Iy$ that comes from the retained points
 in $B(o,y+2\rh) \setminus B(o,y)$.
Hence, Condition 3 holds. 

To prove that Condition 4 holds for the MHP, we use the CCDF of $\xi$  
expressed in the form (15.1.5) in \cite{Daley2}:
\solution
 \begin{align}
F^{\rm c}_{\xi}(x) &= \sum_{k=0}^{\infty} \frac{(-1)^k}{k!} \int_{B(o,x)} \cdots \int_{B(o,x)} \rho^{(k)}(y_1,\ldots,y_k)dy_1\cdots dy_k  \label{egDaley0}\\
		&= \sum_{k=0}^{\infty} \frac{(-1)^k}{k!} \alpha^{(k)}[B(o,x)^{\bigotimes_k}],   
\label{egDaley}
\end{align}
\else
 \begin{align}
F^{\rm c}_{\xi}(x) &= \sum_{k=0}^{\infty} \frac{(-1)^k}{k!} \int_{B(o,x)} \cdots \int_{B(o,x)} \rho^{(k)}(y_1,\ldots,y_k)  \nonumber\\ 
&\quad \cdot dy_1\cdots dy_k  \label{egDaley0}\\
		&= \sum_{k=0}^{\infty} \frac{(-1)^k}{k!} \alpha^{(k)}[B(o,x)^{\bigotimes_k}],   
\label{egDaley}
\end{align}
\fi
where $B(o,x)^{\bigotimes_k}$ is the Cartesian product of $k$ balls and $\alpha^{(k)}$ is the $k$th-order factorial moment measure. For the MHP, the $n$th moment density satisfies 
\begin{equation}
\rho^{(n)}(z_1,\ldots,z_{n}) = \lambda^n, \quad \textrm{ for } (z_1,\ldots,z_{n}) \in S_n(x),
\end{equation}
where $S_n(x) \triangleq \{(z_1,\ldots,z_{n}) \in B(o,x)^{\bigotimes_n}: \|z_i - z_j\| > 2\rh , \forall i \neq j\} $. The complementary set of $S_n(x)$ with respect to $B(o,x)^{\bigotimes_n}$ is $S_n^{\rm c}(x)  = B(o,x)^{\bigotimes_n} \setminus S_n(x)  = \{(z_1,\ldots,z_{n}) \in B(o,x)^{\bigotimes_n}: \exists i \neq j, \textrm{ s.t. } \|z_i - z_j\| \leq 2\rh \}$. 
The Lebesgue measure of $S_n^{\rm c}(x)$ satisfies  $\nu(S_n^{\rm c}(x)) = O(x^{2n-1})$. So, as $x\to \infty$, $\int_{S_n^{\rm c}(x)} \rho^{(n)}(y_1,\ldots,y_n)dy_1\cdots dy_n \to 0$. Since \eqref{egDaley0} can be rewritten as 
\solution
$$F^{\rm c}_{\xi}(x) = \sum_{k=0}^{\infty} \frac{(-1)^k}{k!} \bigg(\int_{S_k(x)} \rho^{(k)}(y_1,\ldots,y_k)dy_1\cdots dy_k + \int_{S_k^{\rm c}(x)} \rho^{(k)}(y_1,\ldots,y_k)dy_1\cdots dy_k \bigg),$$
\else
\begin{align}
F^{\rm c}_{\xi}(x) &= \sum_{k=0}^{\infty} \frac{(-1)^k}{k!} \bigg(\int_{S_k(x)} \rho^{(k)}(y_1,\ldots,y_k)dy_1\cdots dy_k \nonumber \\
&\quad  + \int_{S_k^{\rm c}(x)} \rho^{(k)}(y_1,\ldots,y_k)dy_1\cdots dy_k \bigg),
\end{align}
\fi
it follows that as $x \to \infty$, 
\solution
\begin{align}
F^{\rm c}_{\xi}(x) &\sim \sum_{k=0}^{\infty} \frac{(-1)^k}{k!} \bigg( \int_{S_k(x)} \rho^{(k)}(y_1,\ldots,y_k)dy_1\cdots dy_k +\int_{S_k^{\rm c}(x)} \lambda^{k} dy_1\cdots dy_k \bigg)  \nonumber\\
&=  \sum_{k=0}^{\infty} \frac{(-1)^k}{k!} (\lambda \pi x^2)^k = \exp(-\lambda \pi x^2).  \nonumber
\end{align}
\else
\begin{align}
F^{\rm c}_{\xi}(x) &\sim \sum_{k=0}^{\infty} \frac{(-1)^k}{k!} \bigg( \int_{S_k(x)} \rho^{(k)}(y_1,\ldots,y_k)dy_1\cdots dy_k \nonumber\\ 
&\quad +\int_{S_k^{\rm c}(x)} \lambda^{k} dy_1\cdots dy_k \bigg)  \nonumber\\
&=  \sum_{k=0}^{\infty} \frac{(-1)^k}{k!} (\lambda \pi x^2)^k = \exp(-\lambda \pi x^2).  \nonumber
\end{align}
\fi
Therefore, $\mathbb{E}(\xi^n) $ is bounded for all $n$ and Condition 4 holds.

\end{proof}


\begin{thebibliography}{10}



\bibitem{book}
M.~Haenggi, \emph{Stochastic Geometry for Wireless Networks}, Cambridge University Press, 2012.



\bibitem{JSAC09}
M.~Haenggi, J.~G.~Andrews, F.~Baccelli, O.~Dousse, and M.~Franceschetti, ``Stochastic geometry and random graphs for the analysis and design of wireless networks," \emph{IEEE Journal on Selected Areas in Communications,}  Vol.~27, No.~7, pp.~1029-1046, Sep.~2009.


\bibitem{modelling0}
J.~G.~Andrews, R.~K.~Ganti, M.~Haenggi, N.~Jindal, and S.~Weber, ``A primer on spatial modeling and analysis in wireless networks," \emph{IEEE Communications Magazine,}  Vol.~48, No.~11, pp.~156-163, Nov.~2010. 



\bibitem{Tract}
J.~G.~Andrews, F.~Baccelli, and R.~K.~Ganti, ``A Tractable Approach to Coverage and Rate in Cellular Networks," \emph{IEEE Transactions on Communications}, Vol.~59, No. 11, Nov.~2011.






\bibitem{PPPcellular0}
T.~D.~Novlan, R.~K.~Ganti, A.~Ghosh, and J.~G.~Andrews,  ``Analytical Evaluation of Fractional Frequency Reuse for OFDMA Cellular Networks," \emph{IEEE Transactions on Wireless Communications, } Vol.~10, No.~12, pp.~4294-4305, Dec.~2011.

\bibitem{KTier}
H.~S.~Dhillon, R.~K.~Ganti, F.~Baccelli and J.~G.~Andrews, ``Modeling and Analysis of K-Tier Downlink Heterogeneous Cellular Networks," \emph{IEEE Journal on Selected Areas in Communications}, Vol.~30, No.~3, pp.~550-560, Apr.~2012.


\bibitem{DPPCluster}
W.~C.~Cheung, T.~Q.~S.~Quek, and M.~Kountouris, ``Throughput Optimization, Spectrum Allocation, and Access Control in Two-Tier Femtocell Networks," \emph{IEEE Journal on Selected Areas in Communications,} Vol.~30, No.~3, pp.~561-574, Apr.~2012.


\bibitem{NSCluster}
Y.~Zhong and W.~Zhang, ``Multi-Channel Hybrid Access Femtocells: A Stochastic Geometric Analysis," \emph{IEEE Transactions on Communications,} Vol.~61, No.~7, pp.~3016-3026, Jul.~2013.



\bibitem{PoissonNonCell1}
S.~P.~Weber,  X.~Yang, J.~G.~Andrews, and G.~de Veciana, ``Transmission capacity of wireless ad hoc networks with outage constraints," \emph{IEEE Transactions on Information Theory,}  Vol.~51, No.~12, pp.~4091-4102, Dec.~2005. 

\bibitem{PoissonNonCell2}
S.~Bandyopadhyay and E.~J.~Coyle, ``An energy efficient hierarchical clustering algorithm for wireless sensor networks,"  \emph{Proc. IEEE INFOCOM 2003,} Vol.~3,  pp.~1713-1723, Apr.~2003. 

\bibitem{GantiCluster}
R.~K.~Ganti and M.~Haenggi, ``Interference and outage in clustered wireless ad hoc networks," \emph{IEEE Transactions on Information Theory},  Vol.~55, No.~9, pp.~4067-4086, Sep.~2009.


\bibitem{nonPoisson1}
H.~ElSawy and E.~Hossain, ``A Modified Hard Core Point Process for Analysis of Random CSMA Wireless Networks in General Fading Environments," \emph{IEEE Transactions on Communications,}  Vol.~61, No.~4, pp.~1520-1534, Apr.~2013. 

\bibitem{nonPoisson2}
M.~Haenggi,  ``Mean Interference in Hard-Core Wireless Networks," \emph{IEEE Communications Letters,}  Vol.~15, No.~8, pp.~792-794, Aug.~2011. 


\bibitem{NOW}
M.~Haenggi and R.~K.~Ganti, ``Interference in Large Wireless Networks," \emph{Foundations and Trends in Networking}, Vol.~3, No.~2, pp.~127-248, 2008.



\bibitem{tit11}
R.~K.~Ganti, J.~G.~Andrews, and M.~Haenggi, ``High-SIR Transmission Capacity of Wireless Networks with General Fading and Node Distribution," \emph{IEEE Transactions on Information Theory}, Vol.~57, pp.~3100-3116, May 2011.

\bibitem{High-Reliability}
R.~Giacomelli, R.~K.~Ganti, and M.~Haenggi, ``Outage Probability of General Ad Hoc Networks in the High-Reliability Regime," \emph{IEEE/ACM Transactions on Networking}, Vol.~19, pp. 1151-1163, Aug.~2011.


\bibitem{MIND}
R.~K.~Ganti and M.~Haenggi, ``Interference in Ad Hoc Networks with General Motion-Invariant Node Distributions," \emph{in 2008 IEEE International Symposium on Information Theory (ISIT'08)}, (Toronto, Canada), Jul.~2008.



	

\bibitem{NaTWC14}
N.~Deng, W.~Zhou, and M.~Haenggi, ``The Ginibre Point Process as a Model for Wireless Networks with Repulsion," \emph{IEEE Transactions on Wireless Communications}, 2014. Submitted. Available at \url{http://www3.nd.edu/~mhaenggi/pubs/twc14c.pdf}. 


\bibitem{ModelCWN}
J.~Riihijarvi and P.~Mahonen, ``A spatial statistics approach to characterizing and modeling the structure of cognitive wireless networks," \emph{Ad Hoc Networks}, Vol.~10, No.~5, pp.~858-869, 2012.

\bibitem{AGTWC13}
A.~Guo and M.~Haenggi, ``Spatial Stochastic Models and Metrics for the Structure of Base Stations in Cellular Networks," \emph{IEEE Transactions on Wireless Communications}, Vol.~12, Iss.~11, pp.~5800-5812, Nov.~2013. 


\bibitem{CodingGain}
S.~Lin and D.~J.~Costello, \emph{Error Control Coding}, 2nd ed., Englewood Cliffs, NJ: Prentice-Hall, 2004.





\bibitem{CompositeFading1}
F.~Hansen and F.~I.~Meno, ``Mobile fading-Rayleigh and lognormal superimposed," \emph{IEEE Transactions on Vehicular Technology,}  Vol.~26, No.~4, pp.~332-335, Nov.~1977. 

\bibitem{CompositeFading2}
A.~M.~D.~Turkmani, ``Probability of error for M-branch macroscopic selection diversity," \emph{IEE Proceedings I Communications, Speech and Vision,} Vol.~139, No.~1, pp.~71-78, Feb.~1992.


\bibitem{MeanSINR}
D.~Xenakis, N.~Passas, L.~Merakos, and C.~Verikoukis, ``Energy-efficient and interference-aware handover decision for the LTE-Advanced femtocell network," \emph{IEEE ICC'13}, Jun.~2013.

\bibitem{Tauberian}
K.~Nakagawa, ``Application of Tauberian theorem to the exponential decay of the tail probability of a random variable," \emph{IEEE Transactions on Information Theory}, Vol.~53, pp.~3239-3249, Sep.~2007.

\bibitem{Daley2}
D.~J.~Daley and D.~Vere-Jones, \emph{An Introduction to the Theory of Point Processes: Volume II: General Theory and Structure (Vol.~2)}, Springer, second edition, 2007.


\bibitem{logNormaldef}
S.~Weber and J.G.~Andrews, ``A Stochastic geometry approach to Wideband Ad Hoc Networks with Channel Variations," \emph{Workshop on Spatial Stochastic Models for Wireless Networks,} Apr.~2006.


\end{thebibliography}
\end{document}